\documentclass[oribibl]{llncs}
\usepackage{dsfont}
\usepackage{amssymb}
\usepackage{amstext}
\usepackage{graphicx}
\usepackage{algorithm}
\usepackage{algorithmic}

\usepackage[english]{babel}

\input xy
\xyoption{all}

\usepackage{algorithm}
\usepackage{algorithmic}

 \newcommand{\B}{\mathcal{B}}
 \newcommand{\C}{\mathcal{C}}
 \newcommand{\G}{\mathcal{G}}
 \newcommand{\D}{\mathcal{D}}

 \newcommand{\HH}{\mathcal{H}}

 \newcommand{\comp}{\circ}

\begin{document}

\title{Patterns on data described by vague limits, vague colimits and vague commutativity}

\author{Carlos Leandro\inst{1} \and Lu\'{\i}s Monteiro\inst{2}}

\institute{\email{miguel.melro.leandro@gmail.com}\\ Departamento de Mate\'{a}tica, Instituto Superior de Engenharia de Lisboa, Portugal.
\and CITI, Departamento de Inform\'{a}tica,  Faculdade de Ci\^{e}ncias e Tecnologia, \\Universidade Nova de Lisboa, 2829-516 Caparica, Portugal.
}

\maketitle

\begin{abstract}
The development of machine learning in particular and artificial
intelligent in general has been strongly conditioned by the lack of an appropriated framework to specify and integrate learning processes, data transformation processes and data models. In this work we extend traditional algebraic specification methods to this type of framework. Limits and colimits of diagrams are universal constructions fundamental in different mathematical domains importance in semantic modeling. The aim of our work is to study the possibility of extending these algebraic frameworks to the specification of vague structures and to the description of vague patterns on data.  
\end{abstract}

\section{\uppercase{Introduction}}

Modern human activities impose the description of structures similar to set-theoretic notions, but that are not governed by classical logic. This, in some sense, explains the increasing importance of probabilistic and fuzzy models in our daily life. These models can be seen as patterns that are present on data, and its use is usually governed by probabilistic logic or a fuzzy logic.  We centered our work on the description of vague structures. And given the descriptive power of algebraic tools like sketches, we work on the possibility of performing this description using limits, colimits and commutativity.  For it, this paper presents vague conservative extension to these set-theoretic notions. These notions are described in a general universe for fuzzy modeling given by a class of structures, denoted by $Rel_\Omega$, having by morphisms relations evaluated in a  multi-valued logic $\Omega$, and where composition is defined using a semiring of logic connectives. Objects in this category are characterized by a membership relation and a similarity relation, encoding the degrees of vagueness for ``$x\in X$'' and describing the degree of truth for proposition ``$x=y$.'' In $Rel_\Omega$ morphism are conservative bimodules, a type of relation evaluated in $\Omega$, which conserves membership and similarity degrees between target and source objects.

The universe $Set$ is a substructure of $Rel_\Omega$. In the following we present conservative extensions in $Rel_\Omega$ to the notions of limit, colimit and commutativity in $Set$, in the sense that when a diagram is defined using maps between classic sets, the described extensions coincide with the categorical ones. Furthermore,  our approach allowed extending Ehresmann's sketch structure in two directions. We propose a logic extension, used to specify weighted propositions like ``a distribution $d$ is a vague limite for diagram $D$,'' ``a similarity relation $r$ is a vague colimite to a diagram $D$ colimit,'' and by a ``diagram $D$ is vaguely commutative.'' We also propose a functional extension, where instead of diagrams we use multi-diagrams for the  graphic-based proposition description. Here we assume a multi-graph as a structure defined by arrows linking two sets of vertices. The use of this graphic representation to construction of models impose the existence of a rich interpretation framework, named of \emph{multi-category}. Categories like $Rel_\Omega$ have this nice structure characterized by the existence of operators for the construction of complex objects and morphisms from simplest ones, and the existence of an operator on arrows, defined everywhere having by restriction a composition operator. In this context we see a multi-diagram as a circuit defined by aggregation of multi-morphisms.

Traditional data specification assume the data model to be a correct reflection of the world being captured and assume that the specification accurate. It is rarely the case in real life that these assumptions are met. Where data models are usually vague structures,  generated by data aggregation and characterized by propositions evaluated in nonclassical logics. Different semantics for data modeling have been proposed to handle different categories of data quality (or lack thereof). Our approach to data modeling was centered in the semantic extension of Ehresmann sketches, where objects are assumed to be characterized by membership and similarity relations both evaluated in the same multi-valued logic. This imposed the existence on the modulation universe of an internal multi-valued logic. The categories adequate to this type of modulation were named of $\Omega$-multi-categories, where we can define vague notions of limit and commutativity when the category has local products.  For the description of structures using vague colimits we must assume the existence of an additional additive structure on morphisms. This structure can be find for instance in categories like $Rel_\Omega$ or in the category having as objects families of vectorial spaces where it emerges from the sum of linear transformations.

\section{\uppercase{Multi-categories}}
We begin by presenting the basic notions needed on the definition of multi-diagram and multi-category.

\subsection{Monoidal category on objects}
A category $\C$ is \emph{monoidal}, if there is:
\begin{enumerate}
  \item a bifunctor $\_\otimes\_:\C\times\C\rightarrow \C$, called tensorial product or aggregator,
  \item an object $\top$, called the \emph{unit}, and
  \item for every object $X,Y$ and $Z$ isomorphisms $\alpha_{XYZ}:(X\otimes Y)\otimes Z\cong X\otimes (Y\otimes Z)$, $r_X:X\otimes\top\cong X$ and $l_X:\top\otimes X\cong X$, and $s_{XY}:X\otimes$.
\end{enumerate}
The monoidal category $(\C,\otimes,\top)$ is \emph{monoidal symmetric}, if for every object $X,Y$ there is an isomorphism $Y\cong Y\otimes X$.

A residuum for the tensor $\otimes$, in a monoidal category $(\C,\otimes,\top)$, is a binary operator $\setminus$, defined for its objects such that $X,Y$ and $Z$:
\[
Z=X\otimes Y\text{ se e só se }X=Z\setminus Y.
\]

Note that the residue any not be functorial. A monoidal symmetric category with a functorial residue defines a monoidal close category \cite{maclane71}.

Bellow we described useful monoidal symmetric category used in this work:

\begin{example}[Finite indexed sets]
Consider $Fam$ the category having by objects finite indexed sets $(A_i)_I$, and by morphisms families of maps $(f_j)_J$, is a morphism from $(A_i)_I$ to $(A_j)_J$, if there is a map $\alpha:J\rightarrow I$ such that $f_j:A_{\alpha(j)}\rightarrow A_j$ is a map. $Fam$ is a residuated category, having by object aggregator $(A_i)_I\cup (A_j)_J=(A_i)_{I\amalg J}$, where $I\amalg J$ is the disjoint union, its residuum  is $(A_i)_I\setminus (A_j)_J=(A_k)_K$ a family having by elements indexed sets in $(A_j)_J$ what do not belong to $(A_j)_J$, i.e. $K\subset I$ and $i\in K$ iff for every $j\in J$, $A_j\neq A_i$.
\end{example}

\begin{example}[CRlattice]
A complete residuated lattice (CRlattice for short) is an algebra $\mathbf{\Omega}=(\Omega,\otimes,\Rightarrow,\wedge,\vee,\bot,\top)$ with four binary operations and two constants such that:
\begin{enumerate}
  \item $(\Omega,\wedge,\vee,\bot,\top)$ is a complete lattice with largest element $\top$ and least element $\bot$ (with respect to the lattice ordering $\leq$);
  \item $(\Omega,\otimes,\top)$ is a commutative semigroup with the unit $\top$, i.e. $\otimes$ is commutative, associative and $\top\otimes x = x$ for all $x$;
  \item the residuation equivalence holds:
  \begin{center}
      $z\leq(x\Rightarrow y)$ iff $x\otimes z\leq y$ for all $x,y,z$.
  \end{center}
\end{enumerate}
Since these lattices are complete, for every subset $M\neq \emptyset$ of $\Omega$ we have $\bigvee M\in \Omega$ and $\bigwedge M\in \Omega$.

By a ``many-valued logic'' we mean a logic of which the truth-values set is just a CRlattices or equivalently when the truth-values set is a \emph{quantale} \cite{Stubbe05}. Such a logic is called a monoidal logic in \cite{Hohle94}\cite{Esteve01}.  A CRlattice is a  BL-algebra if additionally the following conditions hold:   $x\wedge y = x\otimes(x\Rightarrow y)$ (is divisibility) and $(x\Rightarrow y)\vee (y \Rightarrow x)= 1$ (is pre-linearity). This kind of logic has been extensively investigated under the name Basic logic in the literature \cite{Hajek98}. Particularly useful BL-algebras, defined when $\Omega$ is the closed unit real interval, when $x\otimes y=\max(x+y-1,0)$ are used for modeling \emph{{\L}ukasiewicz logic}. \emph{G\"{o}del logic} has as models the BL-algebras described using $x\otimes y=\min(x,y)$, and for product logic it is assumed that $x\otimes y=x.y$ (product of reals), see \cite{Hajek98}\cite{Gerla00}.

A CRlattice has a natural structure residuated category, when we take as objects elements in $\Omega$, and if we define a morphism $f:\alpha \rightarrow \beta$ if $\alpha\leq \beta$ in the lattice $\Omega$. The object aggregator is defined by $\otimes$  having by residuum $\setminus$.

On the following, CRlattice structure is importante on the definition of  semirings, used on the definition of composition. With this we will try to catch diferente possibilities for describing the  composition between vague relations.
\end{example}

A functor $F:\C\rightarrow \D$ between resituated categories, is a \emph{strict residuated functor} if it preserves the object agragator, unit and residuum strictly, e.g.
\[
F(X\otimes Y)=FX\otimes FY, F(X\setminus Y)=FX\setminus FY \text{ and } F(\bot)=\bot.
\]

\subsection{Multi-diagrams}

Let $\C$ be a monoidal category on objects with object aggregator $\otimes$. For every object $X$, such that $X\ncong \top$, a \emph{factorization} for $X$ in $\C$ is a family of objects $(X_i)_I$ such that $\bigotimes_{I}X_i= X$,  with each $X_i\ncong \top$. Note that objects may have distinct factorizations. A morphism $f:X\rightarrow Y$ and two factorizations $\bigotimes_{I}X_i= X$ and $\bigotimes_{J}Y_j= Y$, define a \emph{multi-morphism}, denoted in this case by $f:(X_i)_I\rightarrow (Y_j)_J$. We simplify notation by writing $\Box f= (X_i)_I$ and $f \Box =  (Y_j)_J$ and we call them, respectively, $f$ source and target. In Figure \ref{multiarrow} we presented a pictographic representation for a multi-morphism $f$ with $\Box f =\{X_0,X_1,X_2\}$ and $f \Box =\{X_3,X_4,X_5\}$. To this type of structure, linking a set of source nodes and a set of target nodes we called a \emph{multi-arrow}.

\begin{figure}[h]
 \[
 \small
\xymatrix @=5pt {
&&&*+[o][F-]{f}\ar `r[rd][rd]\ar `r[rrd][rrd]\ar `r[rrrd][rrrd]&&&\\
 X_0\ar `u[urrr][urrr]&X_1\ar `u[urr][urr]& X_2\ar `u[ur][ur]&&X_3&X_4&X_5
 }
\]
\caption{Multi-morphism.}\label{multiarrow}
\end{figure}
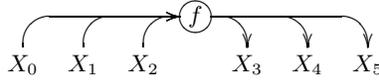

A set of multi-arrows $A$ with nodes on the set $V$ describes a \emph{multi-graph}, represented by a pair $\G=(V,A)$. A source and a target for a multi-graph $\G$ are sets of nodes, denoted respectively by, $\Box \G$ and  $\G \Box$, and we write in this case $\G:\Box \G\rightarrow \G \Box$. By $\bot$ we identify the empty multi-graph having no nodes and no multi-arrows, $\bot=(\emptyset,\emptyset)$.

\begin{definition}[Multi-diagram]
If $\C$ has the structure of a monoidal category on objects, with object aggregator $\otimes$, a \emph{multi-diagram} $D:\G\rightarrow \C$, is a correspondence, $D$, defined from a multi-graph $\G$ to $\C$, assigning to each node in $\G$ a object in $\C$ and to each multi-arrow a multi-morphism defined in $\C$, such that for every multi-arrow $f:\{X_1,\ldots,X_n\}\rightarrow\{Y_1,\ldots,Y_m\}$ in $\G$,
\[
D(f):D(X_1)\otimes\ldots\otimes D(X_n)\rightarrow D(Y_1)\otimes\ldots\otimes D(Y_m).
\]
\end{definition}

On the set of finite multi-graphs we defined a gluing operator. This operator for every pair of multi-graphs $\G_1=(V_1,A_1)$ and $\G_2=(V_2,A_2)$, generates a new multi-graph using as gluing points vertices, with the same label, in the set of target vertices of $\G_1$  and in the set of sources vertices of $\G_2$,  making all other vertices distinct. For multi-graphs $\G_1$ and $\G_2$ it produces a multi-graph $\G_2\circ \G_1$, having by vertices  $V_1\coprod (V_2\setminus(\G_1\Box \cup \G_2\Box))$ and by multi-arrows in $\G_1$ and $\G_2$, $A_1\coprod A_2$.

Similarly to the category freely generated by a digraph \cite{Barr95}, the structure  of every multi-graph $\G$, can be completed to a multi-graph $\G^\ast$, closed for gluing operation: in the sense that every multi-arrow can be seen as a multi-graph and the gluing of two multi-arrows must be a multi-arrow in $\G^\ast$. The multi-graph $\G^\ast$ has the same vertices as $\G$, and has by multi-arrows multi-graphs defined by sets of multi-arrows in $\G$, using the gluing operation. The structure defined by completion having by product the multi-graph gluing operator, and by identity the empty multi-graph $\bot$, denoted $(\G^\ast,\circ,\bot)$, is a monoid. In $\G^\ast$ two multi-arrows $f$  and $g$ are called \emph{composable}  if $\Box g = f \Box$, and the gluing  restricted to composable multi-arrows is a \emph{composition operator}. In this sense when the gluing is restricted to composable multi-graphs $\G^\ast$ is a category. This category has by objects sets of vertices and each object has by identity the empty multi-arrow. The multi-graph $\G^\ast$ is a residuated category where  object aggregator is the set union and having by residuum the set diference. Moreover, for every pair of multi-arrows $f$ and $g$, we have:
\begin{equation}
                \Box(g\circ f)=\Box f \cup \Box g\backslash f\Box,\text{ and }
                (g\circ f)\Box=g\Box \cup f\Box \backslash \Box g.
\end{equation}
This structure presente in $\G^\ast$ can be find in diferente contexts and it will be formalized on the following section.

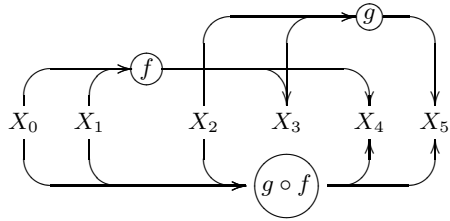
\begin{figure}[h]
\[
\small
\xymatrix @=7pt {
&&&&&*+[o][F-]{g}\ar `r[rdd][rdd] &\\
&&*+[o][F-]{f}\ar `r[rrd][rrd]\ar `r[rrrd][rrrd] &&&&\\
 X_0\ar `u[urr][urr]\ar `d[drrrr][drrrr]&X_1\ar `u[ur][ur]\ar `d[drrr][drrr]& &X_2\ar `u[uurr][uurr]\ar `d[dr][dr]&X_3\ar `u[uur][uur]&X_4&X_5\\
 &&&&*++[o][F-]{g\comp f}\ar `r[ru][ru]\ar `r[rru][rru]&&\\
 }
\]
\caption{A multi-diagram describing the composition of multi-morphisms.}\label{multimorphism composition}
\end{figure}

\subsection{Multi-categories}\label{multicat}
We used essentially multi-categories as a framework for the interpretation of circuits, taken as structures defining  vague relations between a set of input values (the values carried
on its source wires) and a set of output values (the values carried on its target wires). Any two points on the same wire are constrained to carry equal values. We view the discrete components as multi-arrows of a category having by objects sets of types, where we can build complex circuits from the basic components using the gluing operation. Relations have been proposed as a paradigm for circuit development for several reasons. Relations provide a rich algebra for transforming and combining terms, and a natural treatment of non-determinism \cite{Brown95}\cite{Brown94}. The notion of multi-category, proposed in this section, tries to capture the basic structure needed to circuit interpretation.  Here we assume morphisms and objects defined by aggregation of simplest structures. Putting however emphasis on the description of a framework adequate to natural treatment of vagueness. The name multi-category, appears in the literature in different contexts and with different meanings. Here the use of the name multi-category was inspired by the notion of multi-limit propose by Diers on the context of free product completion  \cite{Diers77}. And the described structure is common to many categories, we tried to emphasis the importance of such structure presenting bellow different examples.

A particulary useful example of multi-category can be constructed based on the category of vectorial spaces.

\begin{example}[The multi-category of families of vectorial spaces over $\mathds{K}$]
Let $\mathds{K}$ be a field, the category $Fam(Vec_\mathds{K})$ of families of finite-dimensional spaces over the field $\mathds{K}$, has a structure in some sense similar to the category of multi-graphs. $Fam(Vec_\mathds{K})$ has by objects families of vector spaces $(V_i)_I$ over $\mathds{K}$  and its morphisms are families of linear transformations between vectorial spaces. A morphism $f:(V_i)_I\rightarrow (U_j)_J$ is a  family of linear transformation
\[
f=\left(f_{ij}:V_i\rightarrow U_j\right)_{I\times J}.
\]In this case we define $\Box f=(V_i)_I$ and $f \Box =(U_j)_J$. The composition of linear transformations used on the category of vectorial spaces, can be extended to a total operator in $Fam(Vec_\mathds{K})$. Given two transformations $f:(V_i)_I\rightarrow (V_j)_J$ and $g:(V_k)_K\rightarrow (V_l)_L$, we can define a product between transformations extending composition, we define $h=g\bullet f$, where $f=(f_{ij})_{I\times J}$ and $g=(g_{ij})_{K\times L}$, by: $h_{ij}=f_{ij}$, if $j\notin K\cap J$,  $h_{ij}=g_{jk}\circ f_{ij}$, if $j\in K\cap J$, and $h_{ij}=g_{ij}$ if $i\notin K\cap J$.

Note that, $Fam(Vec_\mathds{K})$ is a residuted category having by object aggregator the union of families, and we have
\begin{equation}
                \Box(g\bullet f)=\Box f \cup \Box g\backslash f\Box,\text{ and }
                (g\bullet f)\Box=g\Box \cup f\Box \backslash \Box g.
\end{equation}
\end{example}

This type structure can also emerge from some known universal completions. Bellow  we presented an example of this using  product completion \cite{Diers77}.

\begin{example}[Free product completion]
The free product completion $\prod(\C)$ of a category $\C$ is a structure having by
\begin{enumerate}
  \item objects small-indexed families $\bar{A}=(A_i)_{i\in I}$ of $\C$-objects $A_i$;
  \item a morphism $\bar{f}:(A_i)_{i\in I}\rightarrow (B_j)_{j\in J}$ in $\prod(\C)$ is given by a function $\varphi:J\rightarrow I$, and by a small-indexed families of $\C$-morphisms $f_j:A_{\varphi(j)}\rightarrow B_j$ $(j\in J)$.
  \item given a morphism $\bar{f}=(f_j:A_{\varphi(j)}\rightarrow A_j)_J$, described by $\varphi:J\rightarrow I$, and a morphism $\bar{g}=(g_j:A_{\psi(j)}\rightarrow A_j)_J$, described by $\psi:I\rightarrow N$, an extension for composition in $\C$ is given defining $\bar{h}=\bar{g}\bullet \bar{f}$, where $\bar{h}=(h_t:A_{\alpha(t)}\rightarrow A_t)_{J\cup I\setminus L}$, with $\alpha:J\cup I\setminus L\rightarrow N\cup L\setminus I$ such that:
        \begin{enumerate}
            \item $\alpha(j)=\varphi(j)$ and $h_j=f_j$, if $\varphi(j)\in L\setminus I$,
            \item $\alpha(j)=(\psi\circ\varphi)(j)$ and $h_j=g_{\varphi(j)}\circ f_j$, if $\varphi(j)\in I\cap L$, and
            \item $\alpha(j)=\psi(j)$ and $h_j=g_j$, if $j \in I\setminus L$.
        \end{enumerate}
\end{enumerate}

In $\prod(\C)$ each family of objects $\bar{A}=(A_i)_I$ has by identity a family $\bar{f}=(1_i:A_i\rightarrow A_i)$ of identity morphisms in $\C$. We define $\Box\bar{f}=(A_i)_I$ and $\bar{f}\Box=(B_j)_J$, if $\bar{f}=(f:A_{\varphi(j)}\rightarrow B_j)_J$ is described by $\varphi:J\rightarrow I$. And we denote by $\top$, the objects defined by an empty family of objects, and by $1_\top\in fam(\C_1)$ the empty family of morphisms.

The completion $\prod(\C)$ has structure of residuated category. The object aggregator can be defined as $(A_i)_I\cup (A_j)_J= (A_k)_{k\in I\coprod J}$,  having by residuum $(A_i)_I\setminus (A_j)_J= (A_k)_K$, if $K\subset I$ and is such that: $i\in K$ iff for every $j\in J$, $A_j\neq A_i$.

The operator $\bullet$ induces a monoidal structure in class of morphisms, having by identity $1_\bot$ and for every pair of arrows $\bar{f},\bar{g}\in fam(\C_2)$ is valid
            \begin{equation}
                \Box(\bar{f}\bullet \bar{g})=\Box \bar{g} \cup \Box \bar{f }\backslash \bar{g}\Box,\text{ and }
                (\bar{f}\bullet \bar{g})\Box=\bar{f}\Box \cup \bar{g}\Box \backslash \Box \bar{f}.
            \end{equation}

Note what, for every pair of objects $\bar{A}=(A_i)_{i\in I}$ and $\bar{B}=(B_j)_{j\in J}$, its product in $\prod(\C)$ is given by $\bar{A}\times \bar{B}=(C_l)_{l\in I\amalg J}$, where $C_l=A_l$ if $l\in I$ and $C_l=B_l$ if $l\in J$, and having projections defined by families of identities described using, respectively, coprojections $p_1:I\rightarrow I\amalg J$ and $p_2:J\rightarrow I\amalg J$.

Let $J_\C:\C\rightarrow set(\C)$ be the canonical embedding, transforming objects of $\C$ in families with a singleton object. When $\C$ has products this embedding defines an isomorphism, given by the product-preserving functor $\Pi:\;\prod(\C)\rightarrow \C$, assigning to each object $\bar{A}$ its product $\prod\bar{A}$ in $\C$, and to each morphism $\bar{f}=(f_j)_J:\Box \bar{f} \rightarrow \bar{f}\Box$,  a morphism $\Pi(\bar{f})=\prod_{j\in J} f_j:\Pi(\Box f)\rightarrow \Pi(f\Box)$.  In this case, we write $\prod(\C)\cong \C$.
\end{example}

Given a category $\C$, the $Free(\C)$ subcategory of $\prod(\C)$ having by objects finite indexed families of $\C$-object and having by morphisms families  $\bar{f}=(f_j:A_{\alpha(j)}\rightarrow B_j)_J$, where $\alpha$ is a bijective map, is known as the free strict monoidal completion for $\C$ \cite{maclane71}.
To the generic structure present on these examples we named \emph{multi-category}, and we define:

\begin{definition}[Multi-category]
A multi-category $\D$, is a residuated category, having by composition $\circ$, where its objects are in a class $\D_0$, with a monoidal structures in the class $\D_1$ of its morphisms. The object agragator is denoted by $\cup$, having by residuum $\setminus$, and by  identity $\bot$. The monoidal structure on multi-arrows, is defined by a associative product operator $\bullet$ having by identity the identify morphism $1_\bot:\bot\rightarrow \bot$. Writing $\Box f$ and $f \Box$, for morphism $f$ source and target, respectively, this two monoidal structures are related, by
\begin{enumerate}
  \item If $f$ and $g$ are composable, i.e. if $\Box f = g \Box$, then $f \bullet g = f \circ g$, and
  \item $\Box(f\bullet g)=\Box g \cup \Box f \backslash g\Box,$ and $(f\bullet g)\Box=f\Box \cup g\Box \backslash \Box f$, for every pair of morphisms $f,g\in \D_1$.
\end{enumerate}
We named multi-morphisms to the multi-category morphisms.
\end{definition}

On the following we simplified notation using $\circ$ to denote both multi-category multi-morphism product and composition operators.

A multi-functor $F:\D\rightarrow \HH$ between multi-categories is a strict residuated functor from $\D$ to $\HH$, preserving the multi-morphisms product,
$F(f\bullet g)=F(f)\bullet F(g).$

Considere  $Free(\C)$ the strict monoidal completion of $\C$, and let $J_c:\C\rightarrow Free(\C)$ be the canonical embedding. Given a multi-category $\HH$ and a functor $F:\C\rightarrow \HH$, there is a unique strict closed functor $\bar{F}:Free(\C)\rightarrow \HH$, such that $\bar{F}\circ J_c=F$. In this sense we see a multi-category as a structural completion for a category.

If for every $\D$-object $A$ there is a family of $\C$-objects $(D_i)_I$ such that $\bigcup_IJ(D_i)\cong A$, multi-category $\D$ is \emph{generated} by $J(\C)$. In this sense we defined:

\begin{definition}\label{multiCatgenerated}
Let $\C$ be a category and $J:\C\rightarrow \D$ an embedding where $\D$ is a multi-category. The multi-category $\D$ is generated by $J(\C)$ if the extension $\bar{J}: Free(\C)\rightarrow \D$ along $J_c:\C\rightarrow Free(\C)$ is surjective on objects.
\end{definition}

And the extension along the canonical embedding can be pushed further:

\begin{proposition}
Consider a category $\C$ and multi-categories $\D$ and $\HH$. Given an embedding $J:\C\rightarrow \D$, such that $\D$ is generated by $J(\C)$. For every functor $F:\C\rightarrow \HH$ there is a unique, up to isomorphism, multi-functor $\overline{F}:\D\rightarrow \HH$ such that $$F=\overline{F}\circ J.$$ And in this case we call to $\D$ a structural completion for $\C$.
\end{proposition}

A category $\C$ has the structure of multi-category if there is a multi-category $\D$ isomorphic to $\C$. By this we mean what, there is an embedding $J:\C\rightarrow\D$ such that every $\D$-object $D$ there is a $\C$-object $C$ such that $J(C)\cong D$.

Since a category $\C$ with products is isomorphic to its free product completion $\prod(\C)$ \cite{Diers77}, e.g. $\C\cong \prod(\C)$, and $\prod(\C)$ has structure of multi-category, we have:

\begin{proposition}
Every category $\C$ with products, has the structure of multi-category.
\end{proposition}

Similarly to categories, the dual of a multi-category $\C$ is a multi-category $\C^{op}$ such that for every multi-morphism $f:\Box f\rightarrow f \Box$ in $\C$ its reverse, $f^\circ: f\Box\rightarrow \Box f$, is  a multi-morphism in $\C^{op}$, and $g^\circ\bullet f\circ^=(f\bullet g)\circ$. It is a natural consequence from multi-category definition that:
\begin{proposition}
If $\C$ is a category with structure of multi-category, then also its dual $\C^{op}$ has the structure of multi-category.
\end{proposition}
In particular, defining $\coprod(\C)$ as the dual of product completion $\prod(\C)$, $\coprod(\C)=(\prod(\C^{op}))^{op}$, if $\C$ has coproducts then $\C\cong \coprod(\C)$ and $\C$ has the structure of multi-category \cite{Adamek94}.

In this sense, for canonical embeddings $J:\C\rightarrow \prod(\C)$ and $J':\C\rightarrow \coprod(\C)$,  in the sense of Definition \ref{multiCatgenerated}, the multi-category $\prod(\C)$ is generated by $J(\C)$ and  the multi-category $\coprod(\C)$ is generated by $J'(\C)$.

Following the spirit proposed by Diers in its extension from limits to multi-limits, described in the context of product completion \cite{Diers77}. When we consider the usual definition of limit, as an initial cone, in the categorical structure of a multi-category, defined using composable morphism \cite{Borceux94}, the \emph{multi-limit} for a diagram $D:\G\rightarrow \C$ in $\C$ is the limit for a diagram in $\C$ completion. When this completion is defined by free product completion, $\D\cong \prod(\C)$, this notion coincide with the Diers' extension for limits.  However this type of extension, by structural completion of the category, is not rich enough to fulfils our needs. In the following sections we will describe an appropriated framework for vague description. For that we need to extend further the notion of structural completion, to motivate that we began by analyzing the multi-category of relations evaluated on a complete resituated lattice.

\begin{example}[Relations evaluated in $\Omega$]
The aim of this work was to present an abstract framework adequate to the description of vague structures. For that, we adopted as reference the multi-category $Rel_\Omega$, of relations evaluated in the multi-valued logic $\Omega$, having its product of multi-morphisms described using a flavor selected in $\Omega$. In the multi-category $Rel_\Omega$ each multi-morphism $f\in Rel_\Omega[A,B]$ can be interpreted as a matrix, having its rows indexed by $A$ and its columns indexed by $B$, with entries in a complete resituated lattice $\mathbf{\Omega}=(\Omega,\otimes,\Rightarrow,\wedge,\vee,\bot,\top)$.

Given the diversity of possible interpretation for ``degree of truth'' in $\Omega$ and of its use for composing relations, we see $Rel_\Omega$ as a class of structures differentiate by the way composition is defined.  A flavor for a multi-category $Rel_\Omega$ is defined by a semiring $(\Omega,\times,\top,+)$, with operations selected in the complete resituated lattice structure $\mathbf{\Omega}$, $+\in\{\oplus,\vee\}$ and $\times\in\{\otimes,\wedge\}$, such that $(\Omega,\times,\top)$ is a monoid, $(\Omega,+)$ is a semigroup and $\times$ distributes over $+$. Flavors are used to differentia ways of relating levels of dependencies between entities.

The order defined in the lattice $\Omega$, can be lifted to each homset in $Rel_\Omega$: for $f,g\in Rel_\Omega[A,B]$, $f\leq g$ iff $f(a,b)\leq g(a,b)$, for every $(a,b)\in A\times B$. Hence each homset $Rel_\Omega[A,B]$ has a top element, denoted by $\top$, such that $\top(a,b)=\top\in \Omega$, and a bottom element denoted by $\bot$, such that $\bot(a,b)=\bot\in \Omega$.

A multi-morphism from the singleton set $\ast$, $\bar{x}\in Rel_\Omega[\ast,A]$ is called a distribution, and it assigns to each $a\in A$, a truth-value $\bar{x}(a)\in \Omega$. In this sense, each endomorphism $f:\ast \rightarrow \ast$ is defined selecting a truth-value, $f(\ast,\ast)=\lambda\in \Omega$. Hence the homset $Rel_\Omega[\ast,\ast]$ is isomorphism to $\Omega$, and the algebraic structure of $\Omega$ can be used to algebrize $Rel_\Omega[\ast,\ast]$ along the isomorphism, denoted by $\ulcorner\_\urcorner:Rel_\Omega[\ast,\ast]\rightarrow \Omega$. We defined an \emph{internal logic} for endomorphisms $f,g\in Rel_\Omega[\ast,\ast]$:
\[
\begin{array}{lll}
  \ulcorner f\otimes g\urcorner=\ulcorner f\urcorner\otimes \ulcorner g\urcorner &\;&  \ulcorner f\vee g\urcorner=\ulcorner f\urcorner\vee \ulcorner g\urcorner\\
  \ulcorner f\wedge g\urcorner=\ulcorner f\urcorner\wedge \ulcorner g\urcorner &\;& \ulcorner f\Rightarrow g\urcorner=\ulcorner f\urcorner\Rightarrow \ulcorner g\urcorner \\
  \ulcorner\top\urcorner = \top &\;& \ulcorner\bot\urcorner = \bot
\end{array}
\]

Note that, each relation $f:A\rightarrow B\in Rel_\Omega$ can be presented as a distribution $f:\ast \rightarrow A\times B$, since the correspondence $f(a,b)=\lambda$ can be encoded using a distribution $\rho_f(\ast,(a,b))=\lambda$. Hence we have
\[
Rel_\Omega[A,B]\cong Rel_\Omega[\ast,A\cup B].
\]
For every multi-morphism $f\in Rel_\Omega[A, B]$ and each $\lambda\in \Omega$, we define an external product $f\downharpoonright \lambda\in Rel_\Omega[A, B]$, by $(f\downharpoonright \lambda)(a,b)=f(a,b)\times \lambda$ and we take by its transposition $f^\circ\in Rel_\Omega[B,A]$ such that $f^\circ(a,b)=f(b,a)$.

When $\Omega$ is a boolean algebra, $Rel_{\{\bot,\top\}}$ is called the multi-category of bivalente relations, and in this case $Rel_{\{\bot,\top\}}[\ast,\ast]\cong \{\bot,\top\}$. $Rel_{\{\bot,\top\}}$ is usually called the category of sets and relations\cite{Borceux94}.

Composition in $Rel_\Omega$ is defined using a selected flavor $(\Omega,\times, \top,+)$.  Given multi-morphisms $f:A\rightarrow B$ and $g:C\rightarrow D$ we define a product in $Rel_\Omega$ by the multi-morphism $$f\circ g:(\Box f \cup \Box g\backslash f \Box)\rightarrow (g \Box \cup f\Box \backslash \Box g),$$ given by the map
\[
(f\circ g)(\bar{x},\bar{z})=\sum_{\bar{y}\in\prod(f\Box\cap \Box g)}f(\bar{x},\bar{y})\times g(\bar{y},\bar{z}),
\]
where $\bar{x}\in\prod(\Box f \cup \Box g\backslash f \Box)$ and $\bar{z}\in\prod(g \Box \cup f\Box \backslash \Box g)$.
\end{example}

Note that, in  $Rel_{\{\bot,\top\}}$, all the possible flavors coincide and the product of composable multi-morphisms is the usual composition of relations.

\section{Logical extension of universal properties}

Similarity is an important concept on definition by approximation. Where the main
goal is to, based on the analyze of data sets, find patterns and regularities
on the data described by structures similar to algebraic structures. In searching for such
regularities, it is usually not enough to consider only equality or inequality
of data elements. Instead, we need to consider how similar, or different two
elements are, i.e. we have to be able to quantify how distinct to two
elements are. This notion is needed in virtually any knowledge discovery application.

How similarity between elements is defined, however, largely depends on
the type of the data. The elements considered in data modeling are often
complex, and they are described by a different number of different kinds of
features. On the other hand, on a single set of data we
can have several kinds of similarity notions.  Different similarity measures can reflect different facets of the data, and therefore, two elements can be determined to be very similar by one measure and very different by another measure. In practice, however similarity degrees have mainly an ordinal meaning. In other words it is the ordering induced by the similarity degrees between the elements that is meaningful, rather than the exact value of the degrees. We assume similarity relations evaluated in a complete lattice. The same set used to describe membership grades of elements to a set, useful on the encoding of data imprecision or uncertainty. This allows the use of membership relations and similarity relations directly for predicate construction. When these relations are evaluated in a multi-valued logic, we called it the \emph{logic of the universe of discurse}.

Despite the fact that there is no single definition for similarity, and that
one single measure seldom suits for every purpose, we try to describe
a generic framework, to the manipulation of objects having similarity and vague membership relations associated.

While our first goal, presented in Section \ref{multicat}, for the definition of multi-categories was essentially functional, as a framework for the relational interpretation for circuits. The idea associated with the notion of $\Omega$-multi-categories is its logical extension, the possibility of internalize in its structure  a multi-valued logic.

In a $\Omega$-multi-categories we assume the existence of an object such that its endomorphisms has a monoidal structure. This tries to capture the structure of $Rel_\Omega$, where for the singleton set $\ast$, $Rel_\Omega[\ast,\ast]$ has by elements endomorphisms $\ulcorner\lambda\urcorner:\ast\rightarrow \ast$ defined by each $\lambda\in \Omega$, used to internalized the logic $\Omega$. Each multi-morphism in $Rel_\Omega$, $f:A\rightarrow B$ is interpreted as a relation $\rho_f:ast\rightarrow A\cup B$, given for every $(a,b)\in A\times B$ as $\rho_f(\ast,(a,b))=f(a,b)$, defining an isomorphism, $Rel_\Omega[A,B]\cong Rel_\Omega[\ast,A\cup B]$.

The structure of a $\Omega$-multi-categories is given by:

\begin{definition}
A multi-category $\D$ is a $\Omega$-multi-category, for a monoid $\mathbf{\Omega}=(\Omega,\times,\top)$, when:
\begin{enumerate}
  \item for every pair of objects $A$ and $B$, a multi-morphism $f\in \D[A,B]$ and scalars $\lambda\in \Omega$ there is an \emph{external product} $f\downharpoonright \lambda\in \D[A,B]$, such that
      \begin{enumerate}
        \item $f\downharpoonright (\lambda\times \alpha)= (f\downharpoonright \lambda)\downharpoonright\alpha$,
        \item $f\downharpoonright \top = f$, and
        \item $(f\circ g)\downharpoonright \lambda = (f\downharpoonright \lambda)\circ g = f\circ (g\downharpoonright \lambda)$;
      \end{enumerate}
  \item for every pair of objects $A$ and $B$, there is \emph{reverse operator} defined using a isomorphism  $$(\_)^\circ: \D[A,B]\rightarrow \D[B,A],$$ such that $(f^\circ)^\circ=f$, $(f\circ g)^\circ=g^\circ\circ f^\circ$ and $1_A^\circ=1_A$;
  \item for every pair of objects $A$ and $B$, $\D[A,B]$ is partially ordered, given two multi-morphisms $f,g:A\rightarrow B$, with $f\leq g$, we have
  \begin{enumerate}
    \item $f^\circ \leq g^\circ$, and
    \item for $h:C\rightarrow A$ and $i:B\rightarrow D$, $i\circ f\circ h\leq i\circ g \circ h$;
  \end{enumerate}
  \item for every pair of objects $A$ and $B$, there is an \emph{operator for tabulation} defining an isomorphism $$\rho_{(\_)}:\D[A,B]\rightarrow \D[\ast,A\cup B].$$
  \end{enumerate}
For every object $A$, $1_A\downharpoonright \_:\;\Omega\rightarrow \D[A,A]$ defines an epimorphism, for the object $\ast$, $1_\ast\downharpoonright \_:\;\Omega\rightarrow \D[\ast,\ast]$ is an isomorphism and its inverse will be denoted by $\ulcorner-\urcorner:\;\D[\ast,\ast]\rightarrow \Omega$ and we have $\ulcorner1_\ast\urcorner=\top$, for every $f,g\in \D[\ast,\ast]$, $\ulcorner f\circ g\urcorner=\ulcorner f \urcorner\times \ulcorner g\urcorner$, $\ulcorner f^\circ\urcorner=\ulcorner f\urcorner$ and $\ulcorner f\downharpoonright \lambda\urcorner = \ulcorner f\urcorner\times \lambda$.
\end{definition}

Given a monoid $(\Omega,\times,\top)$ and a multi-category $\D$, we can generate a $\Omega$-multi-category by structural completion.

\begin{example}[Suszko's completion]
Every monoid $\Omega$ and every multi-category $\D$, with a terminal element $\ast$, can be extended to an $\Omega$-multi-category, denoted by $\Omega(\D)$. For that we weighted formally $\D$-morphisms using values from monoid $\Omega$, defining a new multi-category $\Omega(\D)$ having by objects $\D$-objects and for each $\D$-morphism $f\in\D[A,B]$ and every $\lambda\in \Omega$ we formally define weighted multi-morphisms in $\Omega(\D)$, $(f\downharpoonright\lambda):A\rightarrow B$, $(f\downharpoonright\lambda)^\circ:B\rightarrow A$, $(f\downharpoonright\lambda):\top\rightarrow A\cup B$ and $(f\downharpoonright\lambda)^\circ:\top\rightarrow B\cup A$. The product between multi-morphisms resultes from extending the product in $\D$, by making $$(f\downharpoonright\lambda_0)\comp (g\downharpoonright\lambda_1)=(f\comp g) \downharpoonright(\lambda_0\times \lambda_1).$$ Note what $\D_\Omega[\ast,\ast]=\{(1_\ast\downharpoonright\lambda):\;\lambda\in\Omega\}\cong \Omega$, and we define $\ulcorner(1_\ast\downharpoonright\lambda)\urcorner= \lambda$.

Moreover, the homset $\Omega(\D)[A,B]$ is sorted by $(f\downharpoonright\lambda_0)\leq (g\downharpoonright\lambda_1)$ if $\lambda_0\leq \lambda_1$. Hence for $(f\downharpoonright\lambda_0)\leq (g\downharpoonright\lambda_1)$, $(g\downharpoonright\lambda_1)^\circ\leq (f\downharpoonright\lambda_0)^\circ$ because $\lambda_1^\circ\leq \lambda_0^\circ$. And, if $(f\downharpoonright\lambda_0)\leq (g\downharpoonright\lambda_1)$,  $(i,\lambda_2)\circ(f,\lambda_0)\circ(h,\lambda_3)= (i\circ f\circ h,\lambda_2\times \lambda_0\times \lambda_3)$ and $(i,\lambda_2)\circ(g,\lambda_1)\circ(h,\lambda_3)= (i\circ g\circ h,\lambda_2\times \lambda_1\times \lambda_3)$, when $\lambda_0\leq \lambda_1$, since $\times$ is monotonically increasing $\lambda_2\times \lambda_0\times \lambda_3 \leq \lambda_2\times \lambda_1\times \lambda_3$, then $(i,\lambda_2)\circ(f,\lambda_0)\circ(h,\lambda_3)\leq (i,\lambda_2)\circ(g,\lambda_1)\circ(h,\lambda_3)$.

The functor $J_\Omega:\D\rightarrow \Omega(\D)$ such that $J_\Omega(A)=A$ and $J_\Omega(f)=(f\downharpoonright\top)$ defines a embedding.
\end{example}

A functor $F$ between $\Omega$-multi-categories preserves its structure if it is a multi-functor and if preserves scalar multiplication, i.e.
\[
F(f\downharpoonright \lambda) = F(f)\downharpoonright \lambda.
\]

Consider the canonical embedding on Suszko's completion of $\D$, $J_\Omega:\D\rightarrow \Omega(\D)$, for every functor $F:\D\rightarrow \HH$, where $\HH$ is an $\Omega$-multi-category there is an unique functor, up to isomorphism, $\bar{F}:\Omega(\D)\rightarrow \HH$ which preserves the structure of $\Omega$-multi-category, and such that
\[
\bar{F}\circ J_\Omega = F.
\]
The functor $\bar{F}$ is defined, by $\bar{F}(A)=F(A)$ and for multi-morphisms $f:A\rightarrow B$ and $\lambda \in \Omega$ by
\[
\bar{F}(f\downharpoonright \lambda)=F(f)\times \lambda \in \D[F(A),F(B)].
\]

Note that, if we assume the existence of a functor $G:\Omega(\D)\rightarrow \HH$ in the above conditions, for every $f:A\rightarrow B$ and $\lambda\in \Omega$ we have
\[
G(f\downharpoonright \lambda)=G(f)\downharpoonright \lambda=F(f)\downharpoonright \lambda=\bar{F}(f\downharpoonright \lambda),
\]
then $G=\bar{F}$. In this sense we named $\Omega(\D)$ the \emph{free $\Omega$-multi-category } completion of multi-category $\C$.

A particular useful $\Omega$-multi-category is $Fam(Vect_\mathds{K})$ the multi-category of families of finite-dimensional vector spaces over a field $\mathds{K}=(\mathds{K},\times, 1, +, 0)$.

\begin{example}[$Fam(Vect_\mathds{K})$]
The structure of $\mathds{K}$-multi-category of $Fam(Vect_\mathds{K})$ is induced by the usual external product and the linear transformation transposition. For a multi-morphism between families of vectorial spaces $f:(V_i)_I\rightarrow (U_j)_J$ given by \[
f=\left(f_{ij}\right)_{I\times J},
\]
where each $f_{ij}:V_i\rightarrow V_j$ is a linear transformation, we defined
\[
f\downharpoonright \lambda=\left(f_{ij}\times \lambda\right)_{I\times J}\text{ and }
f^\circ=\left(f_{ji}\times \lambda\right)_{J\times I}.
\]
We have in this case, $\left(f_{ij}\right)_{I\times J}\downharpoonright (\lambda \times \alpha)= \left(f_{ij}\times\lambda \times \alpha \right)_{I\times J}=\left(f_{ij}\times\lambda \right)_{I\times J}\downharpoonright \alpha = (\left(f_{ij} \right)_{I\times J}\downharpoonright \lambda) \downharpoonright \alpha$, $\left(f_{ij}\right)_{I\times J}\downharpoonright 1= \left(f_{ij}\times 1\right)_{I\times J}= \left(f_{ij}\right)_{I\times J}$ and $(\left(f_{ij}\right)_{I\times J}\circ \left(g_{kl}\right)_{K\times L})\downharpoonright \lambda = (\left(f_{ij}\right)_{I\times J}\circ \left(g_{kl}\right)_{K\times L})\times \lambda = (\left(f_{ij}\right)_{I\times J}\times \lambda)\circ \left(g_{kl}\right)_{K\times L}= (\left(f_{ij}\right)_{I\times J}\downharpoonright \lambda)\circ \left(g_{kl}\right)_{K\times L}$.
Naturally we have ${\left(f_{ij}\right)_{I\times J}^\circ}^\circ = \left(f_{ij}\right)_{I\times J}$, $(\left(f_{ij}\right)_{I\times J}\circ \left(g_{kl}\right)_{K\times L})^\circ = \left(g_{kl}\right)_{K\times L}^\circ \circ \left(f_{ij}\right)_{I\times J}^\circ$ and ${\left(1_{ij}\right)_{I\times I}}^\circ= \left(1_{ij}\right)_{I\times I}$. Concerning the homset $Vect_\mathds{K}[(A_i)_I,(B_j)_J]$, assuming the dimension of each vectorial space $A_i$ and $B_j$ are respectively $n_i$ and $m_j$, each linear transformation has a matricial representation $$Vect_\mathds{K}[(A_i)_I,(B_j)_J]\cong Rel_\mathds{K}[n_1\times \ldots \times n_k, m_1\times \ldots \times m_l]$$
and since $$Rel_\mathds{K}[1, n_1\times \ldots \times n_k\times m_1\times \ldots \times m_l]\cong Rel_\mathds{K}[n_1\times \ldots \times n_k, m_1\times \ldots \times m_l],$$ it follows $$Vect_\mathds{K}[(A_i)_I,(B_j)_J]\cong Vect_\mathds{K}[\ast, (A_i)_I\cup(B_j)_J].$$ And naturally, since $dim(\ast)=1$, we have $$Vect_\mathds{K}[\ast,\ast]\cong \mathds{K}.$$
\end{example}

In analogy to this example, in a $\Omega$-multi-category $\D$, we called \emph{scalar} to each morphism in $\D[\ast,\ast]$, denoted in the following as $\lambda_0,\lambda_1,\ldots$ A multi-morphism from $\D[\ast,A]$ or $\D[A,\ast]$ are called \emph{distribution} and usually denoted by $\bar{x}, \bar{y}, \ldots$ and multi-morphisms in $\D[A,B]$ are denoted by letters $f,g,h,\ldots$

\begin{example}[$Rel_\Omega$]
Given an CRlattice $\mathbf{\Omega}=(\Omega,\otimes,\Rightarrow,\wedge,\vee,\bot,\top)$, where we select a flavor, given by a semiring $(\Omega,\times,\top,+)$ used on composition definition in $Rel_\Omega$.

A $\Omega$\emph{-set} is a triple $(A,\bar{x},\alpha)$, denoted as $\bar{x}:\alpha$, with $A$ a set, $\bar{x}:A\rightarrow \Omega$ a distribution defined by a map and $\alpha:A\times A\rightarrow \Omega$ a similarity relation evaluated in $\Omega$, such that $\alpha \circ \bar{x}\leq \alpha$.

If $A=\{(A_i,\bar{x}_i,\alpha_i)\}_I$ and $B=\{(B_j,\bar{y}_j,\beta_j)\}_J$ are sets of $\Omega$-sets then $f:A\rightarrow B$ is a multi-morphism between $\Omega$-sets when it is a map $f:\prod_IA_i\times \prod_JB_j\rightarrow\Omega$, such that $f(a,b)\times(\Pi_I\bar{x}_i)(a)\leq (\Pi_J\bar{y}_j)(b)$, $f(a,b)\times(\Pi_I\alpha_i)(a,c)\leq f(a,c)$ and $(\Pi_J\beta_j)(a,b)\times f(b,c) \leq f(a,c)$. When this is the case we write $\Box f =\{(A_i,\bar{x}_i,\alpha_i)\}_I$ and $f\Box  =\{(B_j,\bar{y}_j,\beta_j)\}_J$.
\[
\small
\xymatrix @=7pt {
\ast \ar[rr]^{\Pi_I\bar{x}_i} \ar[rrdd]_{\Pi_J\bar{y}_j}&&\prod_IA_i\ar[rr]^{\prod_I\alpha_i}\ar[dd]_f&&\prod_IA_i\ar[dd]_f\\
&&&&\\
&&\prod_JB_j\ar[rr]^{\prod_J\beta_j}&&\prod_JB_j\\
 }
\]

In a $\Omega$-set $(A,\bar{x},\alpha)$, when $\bar{x}$ and $\alpha$ are bivalent evaluations, $\bar{x}$ describes a subset of $A$ and $\alpha$ is an equivalent relation. The top element $\top:\ast\rightarrow A$, is defined by $\top(a)=\top$, for every $a\in A$.

Given $\Omega$-sets $\bar{x}:\alpha$ and $\bar{y}:\beta$, a $\Omega$\emph{-map} is a particular type of multi-morphism $f:(\bar{x}:\alpha)\rightarrow (\bar{y}:\beta)$ defined by a map $f:A\rightarrow B$ such that, for each $a,b\in A$, $$\bar{x}(a)\leq \bar{y}(f(a))\text{ and }\alpha(a,b)\leq \beta(f(a),f(b)).$$

When a relation $\bar{x}:\ast\rightarrow A$ is a map between sets, $\bar{x}$ describes the selection of an element in $A$, and we write in this case $!\bar{x}:\top\rightarrow A$ or $!\bar{x}\in A$. In this sense, by $!\bar{x}\in A\times B$ we define the selection of a pair in $A\times B$.

Independently of $Rel_\Omega$ flavor, for every $\Omega$-set $(A,\bar{x},\alpha)$, its identity is the identity map $1_A:A\rightarrow A$ in $A$.
The class of sets of $\Omega$-sets has a monoidal structure defined by set union, diference an having by identity the empty set.

Since $f\in Rel_\Omega[A,B]$ is by definition a map $f:\Pi A\times \Pi B\rightarrow \Omega$ and $$\ast\times\Pi A\times \Pi B\cong \Pi A\times \Pi B,$$ follows $Rel_\Omega[\ast,A\cup B]\cong Rel_\Omega[A,B]$. And for each homset $Rel_\Omega[A,B]$ we have
\[
f\leq g\text{ iff } f(x,y)\leq g(x,y), \text{ for every }x\in A,y\in B.
\]
For every pair of sets $A$ and $B$ the top element $\top$ in set of relations $Rel_\Omega[A,B]$ is the relation given by $\top(x,y)=\top$. And for every relation $f\in Rel_\Omega[A,B]$, we consider its reverse the relation described by $f^\circ(y,x)=f(x,y)$, which defines an isomorphism between $Rel_\Omega[A,B]$ and $Rel_\Omega[B,A]$.

For every multi-morphism $f\in Rel_\Omega[A,B]$ and each $\lambda\in \Omega$, we have by external product $(f\downharpoonright \lambda)(x,y)= f(x,y)\times \lambda$, for every $(x,y)\in A\times B$. Since $f(x,y)\in \Omega$, $\lambda\in \Omega$ and $\Omega$ have structure of semiring
\[
(f\downharpoonright \lambda \times \alpha)=(f\downharpoonright \lambda) \downharpoonright \alpha, \; f\downharpoonright \top = f \text{ and } (f \circ g) \downharpoonright \lambda = (f \downharpoonright \lambda )\circ g = f\circ (g  \downharpoonright \lambda )\]

Moreover, the possibility of encoding every relation evaluated in $\Omega$ as a table describes the isomorphism  $Rel_\Omega[A,B]\cong Rel_\Omega[\ast,A\cup B]$, and justifies the identification of truth values with endomorphisms,  $Rel_\Omega[\ast,\ast]\cong \Omega$.
\end{example}

Note however that, $\Omega$-multi-categories like $Fam(Vect_\mathds{K})$ and $Rel_\Omega$ have also a natural additive structure on its homsets. For every pair of objects $A$ and $B$, and parallel multi-morphisms $f,g\in \D[A,B]$ there is a multi-morphism $f+g\in \D[A,B]$ such that:
\begin{enumerate}
  \item $f\downharpoonright(\lambda+\alpha)=f\downharpoonright\lambda+f\downharpoonright\alpha$,
  \item $(f+g)\downharpoonright \lambda=f\downharpoonright \lambda+g\downharpoonright \lambda$,
  \item $h\circ(f+g)=h\circ f+h\circ g$ if $\Box h = f \Box= g\Box $,
  \item $(f+g)\circ h = f\circ h+g\circ h$ if $\Box f= \Box g = h\Box$,
  \item $(f+g)^\circ=f^\circ+g^\circ$,
  \item $\rho_{(f+g)}=\rho_f+\rho_g$,
  \item if $f\leq g$ and $h\in \D[A,B]$, $f+h\leq g+h$, and
  \item if $f,g\in \D[\ast,\ast]$, then $\ulcorner f+g\urcorner=\ulcorner f \urcorner+\ulcorner g\urcorner$.
\end{enumerate}
The multi-norphism $f+g\in \D[A,B]$ is defined in the multi-category $Fam(Vect_\mathds{K})$ as the sum of linear transformation. On $Rel_\Omega$ it results from extending the additive operator, selected on its flavor, to relations.

\begin{definition}[Additive $\Omega$-multi-category]
We named \emph{additive $\Omega$-multi-category} to a  $\Omega$-multi-category with an additive  operator on multi-morphisms satisfying above conditions.
\end{definition}

Generic $\Omega$-multi-categories can be used to describe logic extensions to the notion of limit, colimt and for diagram commutativity, for that we see the degree of similarity between two multi-morphisms as a relation evaluated on a complete lattice. For that we define what we mean by a $\Omega$-multi-category generated by a category:

\begin{definition}
Let $\D$ be a category, $\Omega$ a monoid and consider $J:\C\rightarrow \D$ an embedding on the $\Omega$-multi-category $\D$. $\D$ is a $\Omega$-multi-category generated by $J(\C)$ if $\D$ is a multi-category generated by $J(\C)$, with the structure of $\Omega$-multi-category.
\end{definition}

An $\Omega$-object, in a $\Omega$-multi-category $\D$, is a triple $(A,\bar{x},\alpha)$ defined using an $\D$-object $A\in\D_0$, a distribution $\bar{x}:\ast\rightarrow A$ and a similarity relation $\alpha:A\rightarrow A$, i.e. a morphism satisfying:
$$1_A\leq \alpha,\;\alpha = \alpha^\circ,\text{ and }\alpha \circ \alpha \leq \alpha.$$
A similarity $\alpha$ is called an \emph{equivalence} when $\alpha \circ \alpha = \alpha$.

\begin{example}
A \emph{generalized metric space} is a set $X$ together with a mapping
\[
d(\_,\_):X\times X\rightarrow \mathds{R}
\]
which satisfy
\begin{enumerate}
  \item $d(x,x)=0$, and
  \item $d(x,z)\leq d(x,y) + d(y,z)$.
\end{enumerate}
 The real number $d(x,y)$ will be called the distance from $x$ to $y$. The pair $(X,d)$ defines a \emph{pseudometric space} if $d$ is a generalized metric such that
 $$d(x,y)=d(y,x).$$
Note that every vectorial space with scalar product is a pseudometric space in particular every euclidian space.

Let $x,y\in \mathds{R}^n$ and let $\langle\cdot,\cdot\rangle$ denote the scalar product in $\mathds{R}^n$. Apart from the linear kernel $k(x,y)=\langle x,y\rangle$ and the normalized linear kernel $k(x,y)=\frac{\langle x,y\rangle}{\|x\|\|y\|}$ the two most frequent kernels on vectorial spaces are the polynomial kernel and the Gaussian RBF kernel\cite{thomas04}. Given two parameters $l\in \mathds{R}$, $p\in \mathds{N}^+$ the polinomial kernel is defined as $k(x,y)=(\langle x,y\rangle+l)^p$ and the Gaussian RBF kernel is defined as $k(x,y)=e^{-l\|x-y\|^2}$.

Positive defined kernel functions can be used on the definition of similarity relations. A simple strategy is define the \emph{distance measured by a kernel} $k$ as
\[
d_k(x,y)=\sqrt{k(x,x)-2k(x,y)+k(y,y)}.
\]
It is a basic result from linear algebra that, if $k$ is positive define, the $d_k$ is a pseudometric.
This allow to use the embedding in a linear feature space by a kernel to define a pseudo-metric for structured data expressed as basic terms \cite{thomas04}. For every pseudo-metric $d(\_,\_):X\times X\rightarrow \mathds{R}$, and every real parameter $l>1$
the map
\[
s_d(x,y)=l^{-d(x,y)}
\]
is a similarity relation.
\end{example}

We simplify notation representing by $\bar{x}:\alpha$ the $\Omega$-object $(A,\bar{x},\alpha)$. We used in the following $\alpha:A$  for denoting the class of distributions in $A$ equipped with the similarity $\alpha$. In this sense the identity $1_A:A\rightarrow A$ defines an equivalence relation in $A$, and a class of multi-morphisms $1_A:A$ having by element for instance $\top:1_A$, given by top distribution $\top:\ast\rightarrow A$.

$\Omega$-objects $\bar{x}:\alpha$ are interpreted as vague structures, where its similarity relation $\alpha:A$ quantifies how identical two elements are,  and  the distribution $\bar{x}$ quantifies the degree of an element belongings to the $\Omega$-object.  $\Omega$-objects are related using bimodules.

In abstract algebra  a bimodule is an abelian group that is both a left and a right module, such that the left and right multiplications are compatible. This notion was extended to enriched categories by B\'{e}nabou using the name of distributor.
Here we adopted the notion of bimodule proposed by Maxwell Kelly on the special case of categories enriched over commutative unital quantale \cite{Kelly82}. Bimodules are relations, defined between  $\Omega$-object in a  $\Omega$-multi-category, preserving the degrees of vagueness in the membership and on similarity. More precisely a morphism $f:(\bar{x}:\alpha)\rightarrow (\bar{y}:\beta)$ between $\Omega$-objects, is a bimodule if it is defined by a morphism $f:A\rightarrow B$ in $\D$, such what
\begin{center}
 $f\circ \bar{x} \leq \bar{y}$,
 $f\circ \alpha\leq f$, and
 $\beta\circ f\leq f$.
\end{center}
The composition in $\D$ is compatible with the structure of a bimodule, since for morphisms $$f:(\bar{x}:\alpha)\rightarrow (\bar{y}:\beta)\text{ and }g:(\bar{y}:\beta)\rightarrow (\bar{z}:\gamma),$$ the morphism $g\circ f$ is a morphism between $\Omega$-objects because, $$g\circ f\circ \bar{x}\leq g \circ \bar{y}\leq \bar{z},\; g\circ f\circ \alpha \leq g \circ f\text{ and }\beta \circ g\circ f \leq g \circ f.$$

Every $\Omega$-object $\bar{x}:\alpha$ has by identity $1_A:A\rightarrow A$ in $\D$. The identity $1_\ast$, defined on object $\ast$ is a similarity relation $1_\ast\leq 1_\ast$, $1_\ast=1_\ast^\circ$ and $1_\ast\circ1_\ast\leq1_\ast$. Defining the $\Omega$-object $(\ast,\top,1_\ast)$ and for every morphism $f\in \D[\ast,\ast]$ defines a multi-morphism $f:(\top:1_\ast)\rightarrow(\top:1_\ast)$.

\begin{definition}[$\D_\Omega$]
The class of $\Omega$-objects and bimodules, with the composition in $\D$, defines a category denoted by $\D_\Omega$.
\end{definition}
In this category we assume that different objects may represent the same entity. For that we defined a order on $\Omega$-objects, having in consideration:

\begin{lemma}
Let $R$ be a bivalente relation defined between $\Omega$-objects, in $\D_\Omega$, such that $(A,\bar{x},\alpha)R(B,b,\beta)$, if there is a morphism $f:(A,\bar{x},\alpha)\rightarrow(B,\bar{y},\beta)$ such that \[f^\circ\circ f= 1_A, f\circ \bar{x} = \bar{y}\text{ and }\beta=f\circ \alpha \circ f^\circ.\]
The relation $R$ is a partial order between $\Omega$-objects.
\end{lemma}
\[
\small
\xymatrix @=7pt {
\ast \ar[rr]^{\bar{x}} \ar[rrdd]_{\bar{y}}&&A\ar[rr]^\alpha\ar[dd]_{f}&&A\ar[dd]_{f}\\
&&&&\\
&&B\ar[rr]^{\beta}&&B\\
 }
\]

Consider in $Rel_{\{\bot,\top\}}$ a bivalent equivalent relation $R$ in $A$, and let $[x]_R$ be the equivalent class of $x$. We will denote the set of equivalent classes as $A/R$. Every $\Omega$-object $(A,\bar{x},R)$ have by refinement $(A/R,[\_]_R\circ\bar{x},1_{A/R})$, this order is defined by the map $[\_]_R:A\rightarrow A/R$, assigning to each element $x$ its equivalence class $[x]_R$, since $1_{A/R}=[\_]_R\circ R \circ[\_]_R^\circ$. We extended this to multi-value logics defining:

\begin{definition}
In a $\Omega$-multi-category $\D$, the $\Omega$-object $(B,b,\beta)$ is a refinement of $(A,a,\alpha)$ if there is a morphism $f:(A,a,\alpha)\rightarrow(B,b,\beta)$ such that \[f^\circ \circ f= 1_A,\;f\circ a = b\text{ and }\beta= f\circ \alpha \circ f^\circ,\]
in this case we write $(B,b,\beta)\leq (A,a,\alpha)$.
\end{definition}

Note that, composition of bimodules is compatible with this order defined for $\Omega$-objects. To show that we assume what two bimodules are equivalent if in some sense describe similar relations between similar $\Omega$-objects. Given $\Omega$-objects $(B,\bar{y},\beta) \leq (A,\bar{x},\alpha)$ and $(B',\bar{w},\beta')\leq (A',\bar{z},\alpha')$, where the congruences are described using bimodules $f:(\bar{x}:\alpha)\rightarrow (\bar{y}:\beta)$ and $f':(\bar{z}:\alpha')\rightarrow (\bar{w}:\beta')$, respectively. For every bimodule $h:(\bar{x}:\alpha)\rightarrow (\bar{z}:\alpha')$ the relation $f'\circ h \circ f^\circ$ is a bimodule between $(\bar{y}:\beta)$ and $(\bar{w}:\beta')$, since it is conservative $f'\circ h \circ f^\circ \circ \bar{y} = f'\circ h \circ \bar{x}\leq f'\circ \bar{z} = \bar{w}$, $f'\circ h \circ f^\circ \circ \beta = f'\circ h \circ f^\circ \circ f \circ \alpha \circ f^\circ = f'\circ h \circ \alpha \circ f^\circ \leq f'\circ h \circ f$ and $\beta' \circ f'\circ h \circ f^\circ =f'\circ \alpha' \circ f'^\circ \circ f'\circ h \circ f^\circ= f'\circ \alpha' \circ h \circ f^\circ\leq f' \circ h \circ f^\circ$. More precisely:

\begin{definition}
Two bimodules $h:(\bar{x}:\alpha)\rightarrow (\bar{z}:\alpha')$ and $t:(\bar{y}:\beta)\rightarrow (\bar{w}:\beta')$ are equivalent if there are bimodules $f:(\bar{x}:\alpha)\rightarrow (\bar{y}:\beta)$ and $g:(\bar{z}:\alpha')\rightarrow (\bar{w}:\beta')$, defining respectively congruences $(B,\bar{y},\beta)\leq (A,\bar{x},\alpha)$ and $(B',\bar{w},\beta')\leq (A',\bar{z},\alpha')$, such that
\[
t=g\circ h\circ f^\circ.
\]
When this is the case we write $t\leq h$.
\end{definition}

Naturally, the composition in $\D_\Omega$ preserves the congruence relation between bimodules, if $h\leq t$ and $g\leq f$ then, when defined, $h\circ g\leq f\circ t$.

The functor $J:\D\rightarrow \D_\Omega$, given by $J(A)=(A,\top,1_A)$ and $J(f)=f$ defines an embedding.  The category $\D_\Omega$ has structure of multi-category, when the object aggregator and its residuum are functorial. They must define for every object $A$ a functor $A\times\_:\D\rightarrow\D$ having by right adjunct $A\setminus\_:\D\rightarrow\D$. In this case e can define where the object aggregator in $\D_\Omega$ and it is $$(A,\bar{x},\alpha)\cup (B,\bar{y},\beta)= (A\cup B,(\bar{x}\cup \bar{y})\circ \top,\alpha\cup \beta),$$ with $A\cup B$ the aggregation of $A$ and $B$ in $\D$, $\alpha\cup \beta:A\cup B\rightarrow A\cup B$ results from parallel morphism aggregation in $\D$, and $(\bar{x}\cup \bar{y})\circ \top$ is the result of composing $\bar{x}\cup \bar{y}:\ast\cup\ast\rightarrow A\cup B$ and $\top:\ast \rightarrow \ast\cup\ast$ the top morphism in $\D[\ast, \ast\cup\ast]$. Every functor $ (A,\bar{x},\alpha)\cup\_:\D_\Omega\rightarrow \D_\Omega$ has by right-adjoint $\_\backslash(A,\bar{x},\alpha) :\D_\Omega\rightarrow \D_\Omega$ given by $(B,\bar{y},\beta)\backslash(A,\bar{x},\alpha) = (B\backslash A ,\bar{x}\backslash A,\alpha\backslash A)$ defined using the right-adjunction to the tensor product $\cup$ in $\D$. The product between bimodules is defined by the natural extension to multi-morphism product in $\D$. For bimodules $f:(A,\bar{x},\alpha)\rightarrow (B,\bar{y},\beta)$ and $g:(C,\bar{z},\gamma)\rightarrow (D,\bar{w},\delta)$, its product $g\circ f$ is defined in $D_\Omega$ as the morphism
\[
g\circ f: (\Box f\cup \Box g \backslash f \Box, \bar{x}', \alpha')\rightarrow (g \Box\cup f \Box \backslash \Box g, \bar{z}', \beta'),
\]
where distributions
\[
\bar{x}':\ast\rightarrow \Box f\cup \Box g \backslash f \Box\text{ and }
\bar{z}':\ast\rightarrow g \Box\cup f \Box \backslash \Box g,
\]
are given by $a'=(\bar{x}\cup \bar{z}\backslash (\Box g \backslash f \Box))\circ \top$ and $b'=(\bar{z}\cup \bar{c}\backslash (f \Box \backslash \Box g))\circ \top$, respectively, and the similarities involved are $\alpha':\Box f\cup \Box g \backslash f \Box$ and $\beta': g \Box\cup f \Box \backslash \Box g$, given respectively by $\alpha'=\alpha\cup \gamma \backslash (\Box g \backslash f \Box)$ and $\beta'=\delta\cup \beta \backslash (f \Box \backslash \Box g)$.

This allows to write:

\begin{proposition}
 Consider a $\Omega$-multi-category $\D$. The category $\D_\Omega$ defined by $\Omega$-objects and conservative bimodules in $\D$ has the structure of multi-category when $\D$ is a closed category, i.e. if the object aggregator and its residuum are functorial. If $\D$ is a multi-category generated by $\C$, then $\D_\Omega$ also is generated by $\C$.
\end{proposition}

Denoting by $\D_\Omega[a:\alpha,b:\beta]$ the class of every multi-morphism $f:(a:\alpha)\rightarrow(b:\beta)$, it is a subclass of $\D_\Omega[\alpha:A,\beta:B]$, the class of every multi-morphism $f:(A,a,\alpha)\rightarrow (B,b,\beta)$.  Then we have $$\D_\Omega[a:\alpha,b:\beta]\subset \D_\Omega[\alpha:A,\beta:B]\subset \D[A,B]. $$

We will use $\D_\Omega$ as our generic framework to describe structures using similarity relations. For that we define:

\begin{definition}\label{similarity}
 In a category with structure of $\Omega$-multi-category, two distributions $\bar{x},\bar{y}:\alpha$ are $\lambda$-similar, with $\lambda\in \Omega$, if $$\ulcorner \bar{y}^\circ \circ \alpha \circ \bar{x}\urcorner = \lambda$$
and in this case we write $[\bar{x}=\bar{y}]_\alpha=\lambda$.
\end{definition}

By definition in a $\Omega$-multi-category every multi-morphism $f:(\bar{x}:\alpha)\rightarrow (\bar{y}:\beta)$ defines a distribution $\rho_f:\alpha\cup \beta$.
\[
\small
  \begin{array}{cc}
\xymatrix @=7pt {
\ast \ar[rr]^{\bar{x}} \ar[rrdd]_{\bar{y}}&&A\ar[rr]^\alpha\ar[dd]_{f}&&A\ar[dd]_{f}\\
&&&&\\
&&B\ar[rr]^{\beta}&&B\\
 }\;\;
&
\small
\;\;
\xymatrix @=7pt {
\ast \ar[rr]^{\rho_f}&&A\cup B\ar[rr]^{\alpha\cup \beta}&&A\cup B\\
 }
  \end{array}
\]
 Considering multi-morphisms $f,g\in \D_\Omega[\alpha:A,\beta:B]$, using Definition \ref{similarity},  we have $$[f=g]_{\alpha\cup\beta}= {\rho_g}^\circ \circ (\alpha\cup \beta) \circ  \rho_f,$$ defining a similarity relation since $[f=f]_{\alpha\cup\beta}= {\rho_f}^\circ \circ (\alpha\cup \beta) \circ  \rho_f\geq  {\rho_f}^\circ \circ  \rho_f \geq 1_\top$, $[f=f]_{\alpha\cup\beta}^\circ= {\rho_f}^\circ \circ (\alpha\cup \beta) \circ  \rho_f= [f=f]_{\alpha\cup\beta}$, $[f=g]_{\alpha\cup\beta}^\circ=( {\rho_g}^\circ \circ (\alpha\cup \beta) \circ \rho_f)^\circ =  {\rho_f}^\circ \circ (\alpha\cup \beta) \circ  \rho_g= [g=f]_{\alpha\cup\beta}$ and $[g=h]_{\alpha\cup\beta}\circ [f=g]_{\alpha\cup\beta} = {\rho_h}^\circ \circ (\alpha\cup \beta) \circ {\rho_g} \circ  {\rho_g}^\circ \circ (\alpha\cup \beta)\circ  \rho_f \leq  {\rho_h}^\circ \circ (\alpha\cup \beta) \circ (\alpha\cup \beta) \circ  {\rho_f} \leq   {\rho_h}^\circ \circ (\alpha\cup \beta) \circ  {\rho_f}=[f=h]_{\alpha\cup\beta}$.

This notion of similarity between multi-morphisms can be seen as a conservative extension to equality in the sense that, morphism equality in $\D_\Omega$ is defined by identity relation, two morphisms $f,g:(\bar{x}:1_A)\rightarrow (\bar{y}:1_B)$ are equal in $\D_\Omega$, $f=g$, iff $[f=g]_{1_{A\cup B}}=\top$ in $\D_\Omega$.

 Using the fact what, by definition in a $\Omega$-multi-category, every homset $\D_\Omega[A,B]$ is partially sorted, following Freyd and Scedrov \cite{Freyd90}, a morphism $f\in \D_\Omega[A,B]$ is called:
\begin{enumerate}
  \item \emph{entire} if $1_A\leq f^\circ\circ f$, and
  \item \emph{simple} if $f\circ f^\circ\leq 1_B$.
\end{enumerate}
A morphism in $\D_\Omega$ is a \emph{map} when it is  entire and simple. When a multi-morphism $f\in \D_\Omega[A,B]$ or a distribution $\bar{x}\in \D_\Omega[\ast,A]$ are defined by maps, we express this by writing $!f: (\bar{x}:\alpha)\rightarrow (\bar{y}:\beta)$ or $!\bar{x}\in A$.

\begin{example}
Let $f: (\bar{x}:\alpha)\rightarrow (\bar{y}:\beta)$ be a morphism $f$ in $Rel_{\{\bot,\top\}}$, then
\begin{enumerate}
  \item $f$ is entire iff $\forall x\exists y: xfy$, and
  \item $f$ is simple iff whenever $xfy$ and $xfz$, we have $y=z$.
\end{enumerate}
\end{example}

In multi-categories like $Rel_\Omega$, a map $!a\in A$, describes the selection of an element in the set $A$. Note also that $\ulcorner\lambda\urcorner:\ast\rightarrow \ast$ is a map iff $\lambda=\top$.

\begin{example}[$Rel_\Omega$]
Consider $Rel_\Omega$ the $\Omega$-multi-category having by flavor $\Omega$ the semiring defined using conjunction $x\otimes y=\max(x+y-1,0)$ and disjunction $x\oplus y=\min(x+y,1)$ from \emph{{\L}ukasiewicz logic}. A distribution $\bar{x}$ in the set $A$ with 4 element can be described using a $4\times 1$ matrix. The distribution $\bar{x}=[1\; 2/3\; 1/3\; 0]^\circ$ is entire $$1_\ast\leq [1\; 2/3\; 1/3\; 0]\circ [1\; 2/3\; 1/3\; 0]^\circ = [1],$$ but fails to be simple  $$[1\; 2/3\; 1/3\; 0]^\circ\circ [1\; 2/3\; 1/3\; 0]=\left[
                                      \begin{array}{cccc}
                                        1 & 2/3 & 1/3 & 0 \\
                                        2/3 & 1/3 & 0 & 0 \\
                                        1/3 & 0 & 0 & 0 \\
                                         0 & 0 & 0 & 0 \\
                                      \end{array}
                                    \right].
$$
\end{example}

In a $\Omega$-multi-category we distinguished a crisp substructure where the computation of categorical definition of limit and colimit take place, assumed governed by the classic bivalent logic.  If $\D$ is a $\Omega$-multi-category we denoted by $\D^\ast$ the subcategory described by all maps in $\D$. $\D^\ast$ has by composition the obvious restriction on the product in $\D$, since the product of composable maps is a map, $!f\circ !g=!(f\circ g)$. In the following the category $\D^\ast$ is called the \emph{crisp subcategory} of $\D$. In particular we have:

\begin{proposition}
For every flavor $\Omega$, $Rel^\ast_\Omega$ the crisp full subcategory of $Rel_\Omega$, is isomorph to $Set$.
\end{proposition}

The interpretation of $\Omega$ as the set of truth values used to govern a $\Omega$-multi-category impose some restriction. This happens when, for instance, we try to modeling data with attributes having by domain structures with distinct, non-isomorph, multi-valued logics $\Omega_0$ and $\Omega_1$. We can show that the category defined by multi-categories and its functors has finite products. More over, if $\C_0$ and $\C_1$ are respectively an $\Omega_0$-multi-category and an $\Omega_1$-multi-category, the product of categories $\C_0\times\C_1$ is an $\Omega_0\times\Omega_1$-multi-category. Since, if $\Omega_0$ and $\Omega_1$ are CRlattices, then the cartesian product $\Omega_0\times \Omega_1$ have a natural structure of CRlattice. In this sense we assume that every logic, associated with each attribute involved on a modeling problem, should be imbedding in a common logic $\Omega$ used on the definition of  a $\Omega$-multi-category used as modulation universe. More precisely, if the problem uses atribules with logics $\Omega_0$ and $\Omega_1$ the modulation universe must be governed by logic  $\Omega_0\times \Omega_1$, where the CRlattice $\Omega_0$ is immersed in $\Omega_0\times \Omega_1$ by the CRlattice homomorphism $h(\lambda)=(\lambda,\top)$.

\subsection{Vague limit}
Given a morphism $f:(\bar{x}:\alpha)\rightarrow(\bar{y}:\beta)$ in $\D_\Omega$, where $\alpha:A$ and $\beta:B$, we represented its evaluation for the pair of distributions $(\bar{x},\bar{y})$ as a distribution with support $A\cup B$ defined by
$$f(\bar{x},\bar{y})=\bar{y}\;^\circ\circ f\circ \bar{x}.$$
When $\bar{x}$ and $\bar{y}$ are maps they select elements in $A$ and $B$ respectively. If $f(!\bar{x},!\bar{y})=\top$ we write, as usual, $f(\bar{x})= \bar{y}$.

\begin{definition}\label{componentes}
Let $\D$ be a $\Omega$-multi-category. When $\D^\ast$ has products, for distributions $!\bar{x}\in A$ and $!\bar{y}\in B$, the unique distribution $\bar{z}:\ast\rightarrow A\times B$, such that $\pi_1\circ \bar{z} = !\bar{x}$ and $\pi_2\circ \bar{z} = !\bar{y}$, is denoted by $\bar{x}\times \bar{y} \in A\times B$.
\end{definition}

Let $\D$ be a $\Omega$-multi-category, and $D:\G\rightarrow\D^\ast$ a diagram, with vertices $(A_i)_I$. The limit of $D$ in $\D^\ast$ is defined as a limit cone $(Lim\;D,(!f_i)_I)$. When this limit exists in $\D^\ast$ we called to $Lim\;D$ the \emph{local limit} of $D$ in $\D$.

Consider now $\D$ a $\Omega$-multi-category with its crisp full subcategory $\D^\ast$ complete. Every diagram $D:\G\rightarrow\D$ with finite vertices $(A_i)_I$, having by limit $(Lim\;D,(!f_i)_I)$ defines a distribution  $\overline{lim}\;D\in \prod_IA_i$. For that, by definition of product, there is a unique morphism $!l:Lim\;D\rightarrow \prod_IA_i$ in $\D^\ast$ such that for every $i\in I$, $\pi_i\circ l=f_i$. Selecting the top distribution $\top\in \D[\ast, Lim\;D]$, we denote by $\overline{lim}\;D$ the distribution defined as $\overline{lim}\;D=!l\circ \top$. For every $!\bar{x}\in Lim\;D$, we can find distributions $!\bar{x}_i\in A_i$, for each $i\in I$, such that $$\pi_i\circ l\circ \bar{x}=\bar{x}_i.$$
When we assign a similarity $\alpha_i$ to each object $A_i$, the morphism $\Pi_I\alpha_i$ is a similarity in $\prod_IA_i$. This similarity defines a $\Omega$-object in $\D_\Omega$ described by the triple $$(\prod_IA_i,\overline{lim}\;D, \Pi_I\alpha_i).$$
This structure can be extended for every cone $(R,(!f_i)_I)$, where $R$ is a $\D$-object and for each $i\in I$, $!f_i:R\rightarrow A_i$ is a map. Let $!l:R\rightarrow \prod_iA_i$ be the unique map such that, for every $i\in I$, $!\pi_i\circ !l= !f_i$. Using the top distribution $\top:\ast\rightarrow R$, we define
\[
F_\top(R,(!f_i)_I)=l\circ \top \in \prod_IA_i,
\]
and the triple $(\prod_IA_i,F_\top(R,(f_i)_I), \Pi_I\alpha_i)$ is an $\Omega$-object in $\D_\Omega$.
\[
\small
\xymatrix @=7pt {
 &&A_i&&\\
 &&&&\\
 R\ar[uurr]^{f_i}\ar[rr]^l&&\prod_iA_i\ar[uu]^{\pi_i}\\
 &&&&\\
 \ast\ar[uurr]_{F_\top(R,(f_i)_I)}\ar[uu]_\top&&&&\\
 }
\]

Let  $\D^\ast$ be a complete category, $D:\G\rightarrow\D$ a diagram with vertices $(A_i,\bar{x}_i,\alpha_i)_I$, and a pair $(R,(f_i)_I)$, with $R=(R,\bar{x},\alpha)$ a $\Omega$-object and multi-morphisms $f_i:(R,\bar{x},\alpha)\rightarrow (A_i,\bar{x}_i,\alpha_i)$ in $\D$. By Definition \ref{similarity} the pair  $(R,(f_i)_I)$ is $\lambda$-similar to $\overline{lim}\; D$, when $$[\overline{lim}\;D\;=\;F_\top(R,(f_i)_I)]_{\prod_i\alpha_i}\geq \lambda.$$

In this framework the limit of a diagram can be extended to the notion of limit of a multi-diagram. For that we must note that, every distribution $\bar{x}\in \prod_{I}A_i$ can be extended to  a distribution $\bar{y}\in \prod_{J}A_j$, with $(A_i)_I\subset (A_j)_J$, given by $\bar{y}=!\pi^\circ \circ \bar{x}$, where $!\pi: \prod_{J}A_j\rightarrow \prod_{I}A_i$ is the obvious projection.
\[
\small
\xymatrix @=7pt {
\ast \ar[rr]^{\bar{x}}&& \prod_{I}A_i \ar[rr]^{\pi^\circ}&& \prod_{J}A_j\\
 }
\]
This type of extension simplifies the use and the manipulation of multi-morphisms, and we called it the \emph{canonical extension} of $\bar{x}$ to $\prod_{J}A_j$.

Considering a multi-diagram in $\D_\Omega$, $D:\G\rightarrow\D_\Omega$. This multi-diagram is defined by a multi-diagram  $D':\G\rightarrow\D$, where we selected for its vertices $(A_i)_I$ distributions $(\bar{x}_i)_I$ and similarities $(\alpha_i)_I$, and each multi-morphism $D'(f):\bigcup_IA_i\rightarrow \bigcup_JA_j$ is assigned to a  bimodule $D(f):\bigcup_I(A_i,\bar{x}_i,\alpha_i)\rightarrow \bigcup_J(A_j,\bar{x}_j,\alpha_j)$ in $\D_\Omega$.

\begin{definition}[Vague limit]\label{weightedlimit}
Let $\D_\Omega$ be a $\Omega$-multi-category with local products and $D:\G\rightarrow\D_\Omega$ a multi-diagram, with vertices $(A_i,\bar{x}_i,\alpha_i)_I$. The vague limit of $D$, is a $\Omega$-object $(\prod_{i\in I}A_i,\overline{lim}\;D,\Pi_{I}\alpha_i)$ defined by distribution $\overline{lim}\;D:\ast\rightarrow\prod_{I}A_i$ given by
\[
\overline{lim}\;D=\Pi_{f\in \G}\overline{\rho_{D(f)}}(\bar{x}),
\]
where $\overline{\rho_{D(f)}}$ is the canonical extension of $\rho_{D(f)}$ to $\prod_{I}A_i$.
\end{definition}

When the multi-diagram $D:\G\rightarrow\D_\Omega$ is a diagram defined using maps, $!\bar{x}\in Lim\;D$ iff $\overline{lim}\;D=\top$.
\[
\small
\xymatrix @=7pt {
 \ast\ar[rr]_{!\bar{x}}\ar[ddrr]_{\overline{lim}\;D}&&Lim\;D\ar[dd]_l\\
 &&&&\\
&&\prod_IA_i\\
 }
\]
In order to measure the quality of a structural approximation using limits we used the notion of similarity.

\begin{definition}\label{similLim}
Given a multi-diagram $D:\G\rightarrow\D_\Omega$ with vertices $(A_i,\bar{x}_i,\alpha_i)_I$, and a pair $(R,(f_i)_I)$, with $(R,\bar{x},\prod_i\alpha_i)$ a $\Omega$-object and multi-morphisms $f_i:(R,r,\alpha)\rightarrow (A_i,\bar{x}_i,\alpha_i)$ in $\D_\Omega$. The pair $(R,(f_i)_I)$ is a $\lambda$-limit of $D$ if $$[\overline{lim}\;D\;=\;F_\top(R,(f_i)_I)]_{\Pi_I\alpha_i}\geq \lambda.$$
\end{definition}

Vague limits in $Rel_\Omega$ can be seen as a logic extension as limits in $Set$ as presented in the following example.

\begin{example}[$Rel_\Omega$]
Every diagram $D:\G\rightarrow Set$, with vertices $(A_i)_I$, defines a multi-diagram in $Rel_\Omega$ using the embedding $J:Set\rightarrow Rel_\Omega$, and defined by $J\circ D:\G\rightarrow Rel_\Omega$, having by vertices $(A_i,\top, 1_{A_i})_I$. The limit cone of $D$ in $Set$,$(Lim\;D,(\alpha_i))_I$, defines a cone in $Rel_\Omega$, having by vertex $(Lim\; D,\top,1_{Lim\;D})$, where $Lim\;D$ can be expressed as a subset of $\prod_I A_i$, see \cite{Borceux94}, given by
\begin{equation}
Lim\;D=\{(\ldots,x_i,\ldots,x_j,\ldots)\in \prod_{I}A_i:\forall_{D(f):A_i\rightarrow A_j}D(f)(x_i)=x_j\}.
\end{equation}\label{limitf}

\[
\small
\xymatrix @=7pt {
\ast \ar[rr]^{\top} \ar[rrdd]_{\overline{lim}\;D}&&Lim\;D\ar[rr]^{1}\ar[dd]_{J(!l)}&&Lim\;D\ar[dd]_{J(!l)}\\
&&&&\\
&&\prod_IA_i\ar[rr]^{\Pi_I\alpha_i}&&\prod_IA_i\\
 }
\]
Consider $!l:Lim\;D\rightarrow \prod_IA_i$ the inclusion in $Set$, the relation $J(!f)\circ \top\in \prod_IA_i$ defines a distribution in $\prod_IA_i$, denoted by $\overline{lim}\;D$. For every distribution $\bar{x}\in \prod_IA_i$, using Definition \ref{similarity}, its similarity with $\overline{lim}\;D$ is given by $[\bar{x}=\overline{lim}\;D]_{\Pi_I\alpha_i},$ when we fixed a similarity relations $\alpha_i$, one for each $A_i$.

In $Rel_\Omega$ we extended the notion of limit in $Set$. According to Definition \ref{weightedlimit}, for every multi-diagram $D:\G\rightarrow Rel_\Omega$ with vertices $(A_i)_I$
\[
(\overline{lim}\;D)(\bar{x})=\Pi_{f\in \G}\overline{\rho_{D(f)}}(\bar{x}),
\]
where the product is computed using $Rel_\Omega$ flavor and each $\overline{D(f)}$ is the canonical extension of $D(f)$ to $\prod_{I}A_i$. A relation $\bar{x}:\ast \rightarrow \prod_{I}A_i$ is the $\lambda$-limit for multi-diagram $D:\G\rightarrow Rel_\Omega$ if $$[\bar{x} = \overline{lim}\;D]\geq \lambda.$$ Then since limits in $Set$ are given by \ref{limitf} we have:

\begin{theorem}
If $D:\G\rightarrow Set$ is a diagram in the category $Set$, with vertices $(A_i)_I$ and limit $L$ then the canonical embedding defines a multi-diagram $J\circ D:\G\rightarrow Rel_\Omega$ and its vague limit is $L$, $\overline{lim}\;(J\circ D)(\bar{x})=\top$ iff $\bar{x}\in L$ i.e.
\[
[L = \overline{lim}\;(J\circ D)]_1=\top.
\]
\end{theorem}

In this sense we see vague limits as a conservative extension in $Rel_\Omega$ to the notion of limit in $Set$.
\end{example}

Bellow we present some common examples of vague limits.
\begin{example}[Vague product]
A discrete diagram $D$ in $Rel_\Omega$ defined using two $\Omega$-objects $(\bar{x}:\alpha)$ and $(\bar{y}:\beta)$, having by support sets $A$ and $B$, respectively, has by weight limit $(A\times B, \bar{x}\times \bar{y}, \alpha\times\beta)$ where
\begin{equation}
(\bar{x}\times \bar{y})(a)= \bar{x}(a)\times \bar{y}(a)\text{ and } (\alpha\times\beta)((a,b),(a',b'))=\alpha(a,a')\times \beta(b,b').
\end{equation}
\end{example}

\begin{example}[Vague equalizers]
A diagram $D$ in $Rel_\Omega$ defined using two parallel morphisms $f,g:(\bar{x}:\alpha)\rightarrow (\bar{y}:\beta)$, with $\alpha:A$ and $\beta:B$ has by limit $(A\times B, \overline{lim}\;D, \alpha\times\beta)$ described by a relation having by support $A\times B$,  given by
\begin{equation}
(\overline{lim}\;D)(a,b)= f(a,b)\times g(a,b).
\end{equation}
\end{example}

\begin{example}[Vague pullback]
A diagram  $D$ in $Rel_\Omega$ defined by $f:(\bar{x}:\alpha)\rightarrow (\bar{z}:\gamma)$, and $g:(\bar{y}:\beta)\rightarrow (\bar{z}:\gamma)$, with $\alpha:A$, $\beta:B$ and $\gamma:C$
has by limit $(A\times B\times C, \overline{lim}\;D,\alpha\times\beta\times\gamma)$ given by
\begin{equation}
(\overline{lim}\;D)(a,b,c)= f(a,c)\times g(b,c).
\end{equation}
\end{example}

\begin{example}\label{example1}
Let $f:A\times B\rightarrow C, g:A\times B\rightarrow C\times D\text{ and } h:A\times C\rightarrow E,$ be morphisms, with supports describe by $\Omega$-sets $\alpha:A, \beta:B, \gamma:C, \delta:D\text{ e } \epsilon:E,$ where $A=B=C=D=\{0,1\}$ and $\alpha=\beta=\gamma=\delta=\epsilon=1_{\{0,1\}\times\{0,1\}}$ the identity relation. If we assume each morphism described by the tables bellow defined in $Rel_{[0,1]}$  when governed by product logic (in this tables the missing cases are assumed to have weighted zero).
\begin{center}
\begin{tabular}{ccc}

$f:\;$\begin{tabular}{|c|c||c||c|}
  \hline
  A & B & C & $\Omega$ \\
  \hline
  1 & 0 & 0 & 1 \\
  0 & 1 & 0 & 1/2 \\
  1 & 1 & 0 & 1/2 \\
  0 & 0 & 0 & 1 \\
  0 & 0 & 1 & 1 \\
  1 & 1 & 1 & 1 \\
  \hline
\end{tabular}&

$g:\;$\begin{tabular}{|c|c||c|c||c|}
  \hline
  A & B & C & D & $\Omega$ \\
  \hline
  0 & 1 & 0 & 1 & 1 \\
  1 & 1 & 0 & 1 & 1/2 \\
  0 & 0 & 0 & 1 & 1 \\
  1 & 1 & 1 & 0 & 1/2 \\
  \hline
\end{tabular}&

$h:\;$\begin{tabular}{|c|c||c||c|}
  \hline
  A & C & E & $\Omega$ \\
  \hline
  1 & 1 & 1 & 1/2 \\
  0 & 0 & 1 & 1 \\
  1 & 0 & 1 & 1/2 \\
  \hline
\end{tabular}
\end{tabular}
\end{center}
Morphisms $f,g,h$ define a diagram $D$, in $Rel_{[0,1]}$, having by limit a relation $\overline{Lim}\;D$ with support $A\times B\times C\times D$ such that, for $(1,0,0,0,1)$ and $(1,1,0,1,1)$, we have respectively,\\
$(\overline{lim}\;D)(1,0,0,0,1)= f(1,0,0)\times g(1,0,0,0)\times h(1,0,1)=1\times 1\times 0\times 1/2 = 0,$ and\\
$(\overline{lim}\;D)(1,1,0,1,1)= f(1,1,0)\times g(1,1,0,1)\times h(1,0,1)=1\times 1/2\times 1/2\times 1/2 = 1/8$. The distribution $\overline{lim}\;D$ can be described by the following table, where the missing cases are assumed to have weight $\bot$.
\begin{center}
\begin{tabular}{c}
$\overline{Lim}\;D:\;$
\begin{tabular}{|c|c|c|c|c||c|}
  \hline
  A & B & C & D & E & $\Omega$ \\
  \hline
  0 & 1 & 0 & 1 & 1 & $1/2\times 1 \times 1= 1/2$ \\
  1 & 1 & 0 & 1 & 1 & $1/2\times 1/2 \times 1/2=1/8$ \\
  0 & 0 & 0 & 1 & 1 & $1\times 1 \times 1=1$ \\
  1 & 1 & 1 & 0 & 1 & $1\times 1/2 \times 1/2=1/4$ \\
  \hline
\end{tabular}\\
\end{tabular}
\end{center}
\end{example}

\subsection{Vague commutativity}\label{def:diagComu}
The commutativity of a diagram $D:\G\rightarrow Set$, in the category of sets, with vertices $(A_i)_I$ can be detected in its tabular internalization $Lim\;D\subset \prod_IA_i$. In this sense we see a limit as a way to encode the diagram structure.   The commutativity of the diagram $D$ given by
\[
\small
\xymatrix @=10pt {
&B\ar[dr]_g&\\
A\ar[ur]_f\ar[rr]_h&&C\\
}
\]
can be expressed  by the equality  $f\circ g=h,$ and when we interprete $f$, $g$ and $h$ as relations in $Rel_{\Omega}$ it is true if and only if, for every $a\in A$, we have
\begin{equation}
\exists_{b\in B, c\in C}f(a,b)\wedge g(b,c)\wedge h(a,c).
\end{equation}
This is equivalente to write, for a select flavor
\begin{equation}
\sum_{b\in B, c\in C}!f(a,b)\times !g(b,c)\times !h(a,c)=\sum_{b\in B, c\in C}(\overline{lim}\;D)(a,b,c)=\top,
\end{equation}
or, using a similarity $\alpha:\Omega$,
\begin{equation}
[\rho_h=\rho_{g\circ f}]_\alpha=\top.
\end{equation}
In this sense, the object $A$ is called \emph{the diagram source} in $D$, the diagram is commutative when
\begin{equation}
\Pi_{a\in A}\sum_{b\in B, c\in C}(\overline{lim}\;D)(a,b,c)=\top.
\end{equation}
Since element selections in $A$ are described by maps $!a:\ast\rightarrow A$, the diagram is commutative if
$$(\pi_{B\times C} \circ \overline{lim}\;D)^\circ\circ !a =\top,\text{ for every map }!a:\ast\rightarrow A$$
\[
\small
\xymatrix @=10pt {
\ast\ar[rr]^{\overline{lim}\;D}\ar[rrdd]_{!a}&&A\times B\times C \ar[dd]^{\pi_{B\times C}}\\
&&\\
&&A\\
}
\]

Generically, a diagram $D$ in $Rel^\ast_\Omega$ is commutative if for every object $A$,  and every element $a\in A$ two sequence of composable maps $f_1,\ldots,f_n$ and $g_1,\ldots,g_m$ from $A$ to $E$ satisfy
\begin{equation}
\exists_{b_1,\ldots,b_{n-1}, c_1,\ldots,c_{m-1},e} f_1(a,b_1)\wedge\ldots\wedge f_n(b_{n-1},e)\wedge g_1(a,c_1)\wedge\ldots\wedge g_m(c_{m-1},e),
\end{equation}
this is equivalente to write, for a select flavor
\begin{equation}
\sum_{b_1,\ldots,b_{n-1}, c_1,\ldots,c_{m-1},e} f_1(a,b_1)\times\ldots\times f_n(b_{n-1},e)\times g_1(a,c_1)\times\ldots\times g_m(c_{m-1},e)=\top,
\end{equation}
or
\begin{equation}
\sum_{b_1,\ldots,b_{n-1}, c_1,\ldots,c_{m-1},e} (\overline{lim}\;D)(a,b_1,\ldots,b_{n-1}, c_1,\ldots,c_{m-1},e)=\top.
\end{equation}

For the conservative extension to the notion of diagram commutativity, of a multi-diagram $D$ in $\D$ a $\Omega$-multi-category with local products, we assume the selection of a set $\Box D $ of its source vertices. Vertices in $\Box D$ are called the \emph{sources} of multi-diagram $D$. The commutativity of $D$ is defined as a relation on those vertices.

Let $I$ and $J$ be two sets of indexes, where $J\subset I$, and a set of objects $\{A_i\}_I$, vertices of a diagram $D$. A projection $\pi_J:\prod_I A_i\rightarrow \prod_{I\setminus J} A_j$ allows the definition of, for every $\Omega$-object $(\prod_I A_i,\bar{x},\alpha)$, an $\Omega$-object $(\prod_{I\setminus J} A_j,\sum_J\bar{x},\sum_J\alpha)$ where
 $\sum_J\bar{x}=\pi_J\circ \bar{x}$, and
 $\sum_J\alpha=\pi_J\circ \alpha\circ \pi_J^\circ$.
Since $\pi_J\circ\pi_J^\circ = 1_{\prod_{I\setminus J} A_j}$, $(\prod_I A_i,\bar{x},\alpha)$ is a refinements for $(\prod_{I\setminus J} A_j,\sum_J\bar{x},\sum_J\alpha)$.
\[
\small
\xymatrix @=7pt {
\ast \ar[rr]^{\bar{\bar{x}}} \ar[rrdd]_{\sum_J \bar{x}}&&\prod_IA_i\ar[rr]^{\alpha}\ar[dd]_{\sum_J}&&\prod_IA_i\ar[dd]_{\sum_J}\\
&&&&\\
&&\prod_{I\setminus J}A_k\ar[rr]^{\sum_J \alpha}&&\prod_{I\setminus J}A_k\\
 }
\]

\begin{definition}[Commutativity]\label{Comutatividade diagramas}
Consider a $\Omega$-multi-category $\D_\Omega$, with local products, and a multi-diagram $D:\G \rightarrow \D_\Omega$ with vertices $(A_i,\bar{x}_i,\alpha_i)_I$. The multi-diagram $D$ is commutative for $\Box D=(A_j,\bar{x}_j,\alpha_j)_J$ with $J\subset I$ if
\[
\pi_J\circ\overline{lim}\;D=\top.
\]
The multi-diagram $D$ is $\lambda$-commutative for $\Box D=(A_j,\bar{x}_j,\alpha_j)_J$, when
\[
\pi_J\circ\overline{lim}\;D\geq \lambda.
\]
A  multi-diagram $D$ is universal if is commutative with empty source, i.e. $\Box D=\emptyset$.
\end{definition}

From the definition of commutativity and by definition of refinement of $\Omega$-objects follows:

\begin{proposition}
Let $\D$ be a $\Omega$-multi-category, with local products,  and $D:\G \rightarrow \D_\Omega$ a multi-diagram with vertices $(A_i,\bar{x}_i,\alpha_i)_I$. The multi-diagram $D$ is commutative for  $\Box D=(A_j,\bar{x}_j,\alpha_j)_J$, with $J\subset I$, if and only if
\[
(\Pi_{I\setminus J} A_k, \top,\sum_J \alpha_j)\leq (\Pi_I A_i, lim\;D,\Pi_I \alpha_i) .
\]
The multi-diagram $D$ is universal when
\[
(\ast, \top,\top)\leq (\Pi_I A_i, lim\;D,\Pi_I \alpha_i).
\]
\end{proposition}

\begin{example}
From Example \ref{example1}, since $(\pi_{B\times C\times D\times E}\circ\overline{lim}\;D)(0)=1/2\vee 1=1$ and  $(\pi_{B\times C\times D\times E}\circ\overline{lim}\;D)(1)=1/8\vee 1/4=1/4$, the multi-diagram is non-commutativity in $\Box D=\{A\}$ for the product logic and when the flavor is $([0,1],\times,1,\vee)$. But it is $1/4$-commutativity in $\{A\}$ since $(\pi_{B\times C\times D\times E}\circ\overline{lim}\;D)(0)\geq 1/4$ and $(\pi_{B\times C\times D\times E}\circ\overline{lim}\;D)(1)\geq 1/4$.
\end{example}

\begin{example} Let $Rel_{[0,1]}$ be governed by the product logic, with flavor $([0,1],\times,1,\vee)$, $\mathds{R}$ be the set of real numbers. In the multi-diagram $D$, presented on Figure \ref{equation1}, each vertices is interpreted as a $[0,1]$-relations $(\mathds{R},\top,1)$,
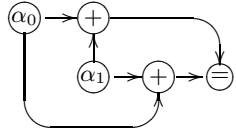
\begin{figure}[h]
\[
\small
\xymatrix @=10pt {
&&&&\\
 *+[o][F-]{\alpha_0}\ar[r]\ar `d[ddrr]`r[rru][drr] &*+[o][F-]{+}\ar `r[rrd][rrd] &&&\\
&*+[o][F-]{\alpha_1}\ar[u]\ar[r] &*+[o][F-]{+} \ar [r]&*+[o][F-]{=}\\
&&&&\\
}
\]
\caption{A multi-diagram encoding $x+y=y+x$.}\label{equation1}
\end{figure}
$=$ representes the relation $=(x,y):=\left\{
                \begin{array}{cc}
                    1 &\text{se } x=y \\
                    0 & \text{se } x\neq y \\
                  \end{array}
                \right.$ and  $+$ is the relation $+:\mathds{R}\times\mathds{R}\rightarrow \mathds{R}$ described using the gaussian   $+(x,y,z)=e^{-\frac{(z-x-y)^2}{2}}.$
Using the notion of vague limit presented on Definition \ref{limit}, for every $x,y,w\in\mathds{R}$ we have,
\[
\begin{array}{rcl}
  (\overline{lim}\;D)(x,y,w) & = & +(x,y,w)\times+(y,x,w')\times = (w,w')\\
                  & = & e^{-\frac{(w-x-y)^2}{2}},\\
\end{array}
\]
since $e^{-(w-x-y)^2}\leq 1$, and because $e^{-(w-x-y)^2}= 1$ when $w=x+y$, follows
\[
 (\pi_{\mathds{R}}\circ\overline{lim}\;D)(x,y) = \bigvee_{w}e^{-(w-x-y)^2}  =  1.
\]
Then the multi-diagram is commutative for  $x$ and $y$ domains. Furthermore, since
\[
 (\pi_{\mathds{R}^3}\circ\overline{lim}\;D)(x,y,z) =  \bigvee_{x,y,w}e^{-(w-x-y)^2} =  1,
\]
the multi-diagram $D$ is universal.
\end{example}

\subsection{Vague colimit}

Besides limits, \emph{colimits} are another important notion for algebraic specification \cite{Diskin99}\cite{piessen00}. It is defined generically as an coequalizer between  coproducts\cite{Borceux94}. In the category of sets this is described  as an equivalente relation defined between set disjoint union. Our extension in a $\Omega$-multi-category presents the colimit as a similarity relation, for that note:

\begin{proposition}
Every map $!f:(\alpha:A)\rightarrow (\beta:B)$ in a $\Omega$-multi-category $\D$, such that $(!f) \circ(!f)^\circ=1_A$, defines an equivalence relation in $A$, $R=(!f)^\circ\circ(!f)$.
\end{proposition}

\begin{proof}
$R\geq 1_B$ because $!f$ is entire, $R^\circ = R$ since $((!f)^\circ\circ(!f))^\circ=(!f)^\circ\circ(!f)$ and $R\circ R = R$ since $(!f)^\circ\circ(!f)\circ(!f)^\circ\circ(!f)=(!f)^\circ\circ(!f)$.
\end{proof}

In particular:

\begin{corollary}
Given a diagram $D:\G\rightarrow\D$, with vertices $(A_i)_I$, if its colimit $(coLim\;D,(!a_i)_I)$ exists in $\D^\ast$ it defines an equivalent in  $\coprod_I A_i$ by the similarity relation$$\overline{colim}\;D:\coprod_I A_i\rightarrow \coprod_I A_i,$$ given by $\overline{colim}\;D=(!j)^\circ\circ(!j)$, where $!j$ is by coproduct definition the only map such that $$!j\circ!p_i=!a_i,$$ for each $i\in I$, where $!p_i$ is  the coprojection in $\D^\ast$.
\end{corollary}

A $\Omega$-multi-category $\D$ has \emph{local coprodutos} if its crisp subcategory $\D^\ast$ has coproducts.

Using $Set$ as a reference, for every diagram $D:\G\rightarrow Set$, with vertices $(A_i)_I$ and morphisms $(f_{i,j,k}:A_i\rightarrow A_j)_{I\times J\times K_{I,J}}$. Without loss of generality, we assume in $\G$ a monoidal structure preserved by $D$. We can assign to $D$ a diagram $D':\G\rightarrow Set$, the \emph{diagram aggregation}, having by vertices $A_i/_{\cong_i}$, sets of equivalence classes for each equivalence relation  $\cong_i$, such that for every $a,a'\in A_i$:
 \begin{enumerate}
   \item $a \cong_i a'$ if exists $j\in J$, $b\in A_j$ and $k,k'\in K_{I,J}$ such that $f_{j,i,k}(b)=a$ and $f_{j,i,k'}(b')=a$, or
   \item $a \cong_i a'$ if exists $j\in J$, $b\in A_j$ and $k,k'\in K_{I,J}$ such that $f_{i,j,k}(a)=f_{i,f,k'}(a')$.
 \end{enumerate}
The family of morfismos $(f_{i,j,k}:A_i\rightarrow A_j)_{k\in K}$ in $D$ defines a morphism in $D'$, $f_{i,j}:A_i/_{\cong_i}\rightarrow A_j/_{\cong_j}$, give for each $k$, by $f_{i,j}([a]_{\cong_i})=[f_{i,j,k}(a)]_{\cong_j}$. Note that, for every $k,k'$, $[f_{i,j,k}(a)]_{\cong_j}=[f_{i,j,k'}(a)]_{\cong_j}$. Furthermore, if $f_{j,l,k}\circ f_{i,j,k'}= f_{i,l,k''}$ then $f_{j,l}\circ f_{i,j}= f_{i,l}$.

Let $(coLim\;D,(!a_i)_I)$ be a colimit  for $D:\G\rightarrow Set$ and $(coLim\;D',(!b_i)_I)$ the colimit of $D':\G\rightarrow Set$. For every set $A_i$, the equivalente relation $\cong_i$ defines a map $[\_]_{\cong_i}:A_i\rightarrow A_i/_{\cong_i}$ used on definition of a cocone $(coLim\;D',(!b_i\circ [\_]_{\cong_i})_I)$ on the diagram $D$, and by colimit definition there is a unique  map $[\_]:coLim\;D\rightarrow coLim\;D'$, such that $[\_]\circ a_i = b_i\circ [\_]_{\cong_i}$. Assuming the possibility of selecting a representative element in each equivalence class, we define the map $(\_)_{\cong_i}: A_i/_{\cong_i}\rightarrow A_i$, such that $([a])_{\cong_i}=a$. This morphism allows the definition of a cone $(!a_i\circ (\_)_{\cong_i})_I$ on the diagram $D'$, from which there is a unique factorization $(\_):coLim\;D\rightarrow coLim\;D'$, such that $a_i\circ (\_) = [\_]_{\cong_i} \circ b_i$. This implies the existente of a isomorphism between $coLim\;D$ and $coLim\;D'$.

In this sense the colimit $(coLim\;D,(!a_i)_I)$  on the diagram $D:\G\rightarrow Set$, can be defined using the isomorphism between $coLim\;D$ and $(\coprod_I(A_i/_{\cong_i}))/r$, where $r:\coprod_I(A_i/_{\cong_i})\rightarrow \coprod_I(A_i/_{\cong_i})$, is described by disjoint union of morphisms in $D':\G\rightarrow Set$, i.e. $$r = \coprod_{  I\times I}(f_{i,j}\vee f^\circ_{j,i}).$$ We see  equivalent relations $\cong_i$ as a mechanism for element aggregation in each vertices,  and the relation $r$ defined in  $\coprod_I(A_i/_{\cong_i})$ as a way for aggregation of elements in distinct vertices. We simplified the use of $r$ defining it by block decomposition, witting for that $r_{ij}=f_{ij}\vee f^\circ_{j,i}$.

The embedding of $Set$ in a $\Omega$-multi-category $Rel_\Omega$, $Set\cong Rel_\Omega^\ast$, for every diagram $D:\G\rightarrow Rel_\Omega^\ast$, with vertices in $(A_i)_I$ and morphisms $(f_{i,j,k}:A_i\rightarrow A_j)_{I\times J\times K_{I,J}}$. The colimit $coLim\;D$ can be constructed using diagram $D':\G'\rightarrow Rel_\Omega^\ast$ and described in the set $(\coprod_I(A_i/_{\cong_i}))/r$, where the relation $r:\coprod_I(A_i/_{\cong_i})\rightarrow \coprod_I(A_i/_{\cong_i})$ is described by
$$r(a,b) = \sum_{  I\times I}(f_{i,j}(a,b)+f^\circ_{j,i}(a,b)),$$
assuming $f_{i,j}(a,b)=\bot$ that every $a\notin A_i/_{\cong_i}$ or $b\notin A_j/_{\cong_j}$. Note that, in this conditions, $r$ is an equivalente relation: since $r_{ii}=1$, then $r\geq 1$, $r_{ij}^\circ=(f_{ij}+f_{ji}^\circ)^\circ=f_{ij}^\circ+f_{ji}=r_{ji}$ e $r_{ji}\circ r_{kj}=(f_{ij}+f_{ji}^\circ)\circ(f_{kj}+f_{jk}^\circ)=f_{ij}\circ f_{kj}+ f_{ji}^\circ\circ f_{kj}+f_{ji}\circ f_{jk}^\circ+f_{ji}^\circ\circ f_{jk}^\circ=f_{ki}+f_{ki}+f_{ik}^\circ+f_{ik}^\circ=f_{ki}+f_{ik}^\circ=r_{ki}$.

In order to present the notion of vague colimit on the multi-diagram $D:\G\rightarrow\D_\Omega$, we must restrict the $\Omega$-multi-categoria $\D$ structure. We must assume that $\D$ is an additive $\Omega$-multi-categoria with local coproduts. In $\D$ we also impose that for every multi-morphism $f:A\rightarrow B$, $f+f=f$. This implies $\lambda+\lambda=\lambda$ hold for every element $\lambda$ in the semi-ring $(\Omega,\times, 1, +)$. Example for this are the $\Omega$-multi-categories defined using G\"{o}del's logic.

With no loss of generality, vague colimit are computed for multi-diagrams with a wick monoidal constrain. We assume that multi-diagrams  $D:\G\rightarrow\D_\Omega$, with vertices $(A_i,\bar{x}_i,\alpha_i)_I$ and multi-morphisms $(f_{J,L,k}:(A_j)_J\rightarrow (A_l)_L)_{I\times I\times K_{I,L}}$, where $J,L\subset I$, satisfying
\begin{enumerate}
  \item If $f_{J,L,k}$ and $f_{J',L',k'}$ are multi-morphisms on $D$, with $J'\subset L$, there is, on $D$, a  multi-morphism $f_{J\cup I\setminus J',L'\cup J'\setminus I,k''}$ such that $$f_{J',L',k'}\circ f_{J,L,k} \leq f_{J\cup I\setminus J',L'\cup J'\setminus I,k''},$$
  \item For every $J\subset I$ there is a multi-morphism $f_{J,J,k}=\sum_J\alpha_j$, defined by similarity relations on vertices of $D$ indexed by $J$.
\end{enumerate}

For a multi-diagram in the above condition, $D:\G\rightarrow\D_\Omega$, its aggregation is a multi-diagram $D_{[]}:\G\rightarrow\D_\Omega$, the  multi-diagram having by vertices families of source and target vertices for multi-morphisms in $D$. And such that, between two families $(A_j)_J$ and $(A_l)_L$ there is a unique multi-morphism defined by $$f_{J,L}=\sum_{k\in K_{J,L}} f_{J,L,k}.$$ We simplify notation by denoting by $B_{J}$ the vertices of $D_{[]}$ defined by the family $(A_j)_J$, if there is a multi-morphism $f_{J,L}$ or $f_{L,J}$, and write $\B$ for the set of all this vertices.

The multi-morphism $c:\coprod \B\rightarrow \coprod \B$ defined as
\[
c=\coprod_{B_J,B_L} (f_{J,L}+f_{L,J}^\circ),
\]
is a similarity relation in $\coprod \B$, since
$c_{II}=f_{II}\geq 1_{B_I}$, $c_{IJ}^\circ=(f_{J,L}+f_{L,J}^\circ)^\circ=f_{J,L}^\circ+f_{L,J}=c_{JI}$. Because in $\D$ we have $f_{J,L,k}\circ f_{J',L',k'}\leq f_{J\cup I\setminus J',L'\cup J'\setminus I,k''}$, it follows $f_{J,L}\circ f_{J',L'}\leq f_{J\cup I\setminus J',L'\cup J'\setminus I}$, then $c_{L,J}\circ c_{J,L'}=(f_{L,J}+f_{J,L}^\circ)\circ(f_{J,L'}+f_{L',J}^\circ)\leq f_{L,L'}+f_{L',L}^\circ=c_{L,L'}$.

\begin{definition}[Vague colimits]
Let $\D$ be an additive $\Omega$-multi-category,  such that in the semi-ring $(\Omega,\times, 1, +)$, holds $\lambda+\lambda=\lambda$.  Consider a multi-diagram $D:\G\rightarrow\D_\Omega$ having by vertices $(A_i,\bar{x}_i,\alpha_i)_I$ and a  multi-morphism $(f_{J,L,k})_{I\times I\times K_{I,L}}$, and for every $J,L'\subset I$, $L\subset I$, $J'\subset L$ and $k,k'\in K$, there is $k''\in K$ such that $$ f_{J',L',k'}\circ f_{J,L,k}\leq f_{J\cup I\setminus J',L'\cup J'\setminus I,k''}.$$ The vague colimit  for diagram $D:\G\rightarrow\D_\Omega$, is computed using the aggregation diagram $D_{[]}:\G\rightarrow\D_\Omega$, and it is a $\Omega$-object defined as $(\coprod\B,\alpha_{\B},\overline{colim})$, where the similarity relation $\overline{colim}\;D:\coprod \B\rightarrow \coprod \B$ is described using blocks $\overline{colim}_{J,L}(D)$, by
\[
\overline{colim}\;D=\coprod_{B_J,B_L\subset \B}\overline{colim}_{J,L}(D).
\]
For each used block we have $\overline{colim}_{J,L}(D)=f_{J,L}+f_{L,J}^\circ$.
\end{definition}

Bellow we present some common examples of vague colimits.

\begin{example}[Vague product]
A discrete diagram $D$ in $Rel_\Omega$ defined using two $\Omega$-objects $(\bar{x}:\alpha)$ and $(\bar{y}:\beta)$, having by support sets $A$ and $B$, respectively, has by vague colimit $(A\coprod B, \bar{x}\coprod \bar{y}, \alpha+\beta)$ where
\begin{equation}
(\alpha+\beta)(a,b)=\alpha(a,b)+\alpha(b,a)+\beta(a,b)+\beta(b,a)=\alpha(a,b)+\beta(a,b),
\end{equation}
where we assume $\alpha(a,b)=\bot$ and $\beta(a,b)=\bot$ when this relations are not defined for $(a,b)$.
\end{example}

\begin{example}[Vague coequalizers]
A diagram $D$ in $Rel_\Omega$ defined using two parallel morphisms $f,g:(\bar{x}:\alpha)\rightarrow (\bar{y}:\beta)$, with $\alpha:A$ and $\beta:B$ has by vague colimit $(A\coprod B, \bar{x}\coprod \bar{y},\overline{colim}\;D)$ described by a relation having by support $A\coprod B$,  given when $a,b\in A\coprod B$
\begin{equation}
(\overline{colim}\;D)(a,b)= f(a,b) + f(b,a) +g(a,b)+g(b,a) + \alpha(a,b) + \beta(a,b),
\end{equation}
where we assume $f(a,b)=\bot$, $g(a,b)=\bot$, $\alpha(a,b)=\bot$ and $\beta(a,b)=\bot$ when this relations are not defined.
\end{example}

\begin{example}[Vague pushout]
A diagram  $D$ in $Rel_\Omega$ defined by $f:(\bar{z}:\gamma)\rightarrow  (\bar{x}:\alpha)$, and $g: (\bar{z}:\gamma)\rightarrow(\bar{y}:\beta)$, with $\alpha:A$, $\beta:B$ and $\gamma:C$
has by vague colimit $(A\coprod B\coprod C, \bar{x}\coprod \bar{y}\coprod \bar{z},\overline{lim}\;D)$ given by
\begin{equation}
(\overline{colim}\;D)(a,b)= f(a,b) +f(b,a)+ g(a,b)+ g(b,a) + \alpha(a,b) + \beta(a,b)+ \gamma(a,b),
\end{equation}
when $a,b\in A\coprod B \coprod C$ and where we assume $f(a,b)=\bot$, $g(a,b)=\bot$, $\alpha(a,b)=\bot$, $\beta(a,b)=\bot$ and $\gamma(a,b)=\bot$ when this relations are not defined.
\end{example}

\section{Pattern in a data set}
A \emph{data set} is a table or a weighted table, in Table \ref{dataset} we can find an example. In this sense the limit of a multi-diagram can be seen as a data set, and we named the diagram a \emph{description} for the data set. When a domain is governed by a multi-valued logic, we are interested in problems having its information encoded on data sets, described using sets of diagrams in a $\Omega$-multi-category $Rel_\Omega$. A diagram $D:\G\rightarrow Rel_\Omega$, in this case, is usually called an \emph{instantiation} for the structure specified by $\G$ \cite{Barr95}.

The vertices $(A_i,\bar{x}_i,\alpha_i)_I$ of an instantiation $D:\G\rightarrow Rel_\Omega$ are usually named \emph{attributes}. For every multi-diagram homomorphism $Q:\G_0\rightarrow \G$, the limit  $\overline{lim}\;(D\circ Q)\in \prod_JA_j$ describes a data set with structure $D\circ Q$. Following the usual approach on sketches theory \cite{Barr95}, a homomorphism $Q:\G_0\rightarrow \G$ is called a \emph{query} to data structure $D:\G\rightarrow Rel_\Omega$. The distribution defined by the limit $\overline{lim}\;(D\circ Q)\in \prod_JA_j$ is interpreted as the \emph{answer} to the query  $Q:\G_0\rightarrow \G$ in data structure $D:\G\rightarrow Rel_\Omega$.

In this context it seems natural to assume that a diagram $D_e:\G_e\rightarrow Rel_\Omega$, having by vertices $(A_i,\bar{x}_i,\alpha_i)_I$, describes a structure in a $\Omega$-object $(S,\bar{s},\beta)$, having by distribution $\bar{s}\in S$, where $S\subset\prod_JA_j$, with $I\subset J$, if there is a map
$!i:S\rightarrow \prod_IA_i$, such that $i^\circ\circ i= 1_S$, $\bar{s}=i^\circ\circ \overline{lim}\;D_e$, and $\beta=i^\circ\circ\prod_I\alpha_i\circ i$, i.e. if
\[
(\prod_IA_i,\overline{lim}\;D_e, \prod_I\alpha_i)\leq (S,d,\beta).
\]
However this notion of description is very restrictive, a more useful notion can be presente by approximation.

\begin{definition}[$\lambda$-description]
The diagram $D_e:\D_e\rightarrow Rel_\Omega$  having by vertices $(A_i,\bar{x}_i,\alpha_i)_I$,  is a $\lambda$-description for the structure in a $\Omega$-object $(S,\bar{s},\beta)$, having by distribution $\bar{s}\in S$, where $S\subset\prod_JA_j$ and $I\subset J$, if there is a map
$!i:S\rightarrow \prod_IA_i$, such that $i^\circ\circ i= 1_S$, and
\[
[\bar{s}=i^\circ.\overline{lim}\;D_e]_\beta\geq\lambda.
\]
\end{definition}

\begin{example} In knowledge representation usually models are expressed using functional components, like logic connectives. A particulary important  methodology is the representation of knowledge using artificial neural network. Where the knowledge is described by a net of processing units (artificial neurons) linked together \cite{Bishop96}. Usually this structures are generated automatic.

Neural networks can be seen as multi-diagrams, where each multi-arrow represents a processing unit (a functional dependence) interpreted as a map in $Rel_{[0,1]}$. A particularly useful type of neural network are the \emph{{\L}ukasiewicz neural network} ({\L}NN) described in \cite{Castro98}. In this type of neural networks processing can be parameterized in order to describe formulas from propositional {\L}ukasiewicz logic \emph{{\L}L}. In Table \ref{semioticaLuk} we can see the correspondence between some formulas and neuronal parametrization. Each processing unit in a {\L}NN with $n$ input wires and one output is interpreted in $Rel_{[0,1]}$ as a map
\[
z= \psi_b(w_1x_1,w_2x_2,\ldots,w_nx_n) =\min(1,\max(0,w_1x_1+w_2x_2+\ldots+w_nx_n+b)),
\]
where parameters $w_1,w_2,\ldots,w_n$ are usually named heights and $b$ the neuron bias. Every formula in {\L}L can be codified using a {\L}NN having by heights in the set $\{1,0,-1\}$ and by bias an integer. In Table \ref{semioticaLuk} we can identify {\L}L connectives using disjunctive and conjunctive formals. Furthermore, it is a simple task identify when a processing unit in a {\L}NN codify a disjunctive or conjunctive formula, for it we used the following result:

\begin{proposition}\cite{Leandro09}\label{conf classification}
Given the neuron configuration
$$
\alpha=\psi_b(-x_1,-x_2,\ldots,-x_n, x_{n+1},\ldots,x_m),
$$
with $m=n+p$ inputs and where $n$ and $p$ are, respectively, the number of negative and the number of positive weights, on the neuron configuration:
\begin{enumerate}
  \item If $b=-p+1$ the neuron is \emph{conjunctive} and it is interpretable as
$$
\neg x_1\otimes\ldots\otimes\neg x_n\otimes x_{n+1}\otimes\ldots\otimes x_m.
$$
  \item When $b=n$ the neuron is \emph{disjunctive} and it is interpretable as
$$
\neg x_1\oplus\ldots\oplus\neg x_n\oplus x_{n+1}\oplus\ldots\oplus x_m.
$$
\end{enumerate}
\end{proposition}

\begin{table}
\begin{center}
\tiny
\begin{tabular}{|c|c||c|c||c|c||c|c|}
  \hline
  Formula: & Configuration: & Formula: & Configuration: & Formula: & Configuration: & Formula: & Configuration: \\
  \hline
  \hline
  $\neg x\oplus y$
  &
  $\xymatrix @R=6pt @C=6pt { x\ar[dr]_{-1} & \ar@{-}[d]^{1} & \\
                           & *+[o][F-]{\varphi} \ar[r]& \\
           y\ar[ur]^{1} &  &  \
           }$
  &
  $ x\otimes \neg y$
  &
  $\xymatrix @R=6pt @C=6pt { x\ar[dr]_{1} & \ar@{-}[d]^{0} & \\
                           & *+[o][F-]{\varphi} \ar[r]& \\
           y\ar[ur]^{-1} &  &  \
           }$
  &
  $x\oplus y$
  &
  $\xymatrix @R=6pt @C=6pt { x\ar[dr]_{1} & \ar@{-}[d]^{0} & \\
                           & *+[o][F-]{\varphi} \ar[r]& \\
           y\ar[ur]^{1} &  &  \
           }$
  &
  $\neg x\otimes \neg y$
  &
  $\xymatrix @R=6pt @C=6pt { x\ar[dr]_{-1} & \ar@{-}[d]^{1} & \\
                          & *+[o][F-]{\varphi} \ar[r]& \\
           y\ar[ur]^{-1} &  &  \
           }$ \\
  \hline
  $x\oplus \neg y$ & $\xymatrix @R=6pt @C=6pt { x\ar[dr]_{1} & \ar@{-}[d]^{1} & \\
                           & *+[o][F-]{\varphi} \ar[r]& \\
           y\ar[ur]^{-1} &  &  \
           }$ & $x\otimes y$ & $\xymatrix @R=6pt @C=6pt { x\ar[dr]_{1} & \ar@{-}[d]^{-1} & \\
                           & *+[o][F-]{\varphi} \ar[r]& \\
           y\ar[ur]^{1} &  &  \
           }$ &
  $\neg x\otimes y$ & $\xymatrix @R=6pt @C=6pt { x\ar[dr]_{-1} & \ar@{-}[d]^{0} & \\
                           & *+[o][F-]{\varphi} \ar[r]& \\
           y\ar[ur]^{1} &  &  \
           }$ &
           $\neg x \oplus \neg y$
             &
             $\xymatrix @R=6pt @C=6pt { x\ar[dr]_{-1} & \ar@{-}[d]^{2} & \\
                           & *+[o][F-]{\varphi} \ar[r]& \\
           y\ar[ur]^{-1} &  &  \
           }$
               \\
           \hline
\end{tabular}
\end{center}
\caption{Possible configurations for a neuron in a {\L}NN a its interpretation.}\label{semioticaLuk}
\end{table}
A topology for a neuronal network can be described by a multi-diagram constructed by the selection of multi-arrows and interpreting them as neuron configurations. In this case every multi-morphisms used on the multi-diagram is a map and every wire links vertices of the some type, $\Omega=[0,1]$, in $Rel_{[0,1]}$. If $D:\G\rightarrow Rel_{[0,1]}$ is a diagram describing a {\L}NN, its vague limit is a distribution $\overline{lim}\;\D\subset \prod_I\Omega$, where the finte set $I$ of indexes is defined by wires used in the diagram. In this sense, $\prod_I\Omega=\Omega^n$ where $n$ is the number of wires in $D$.
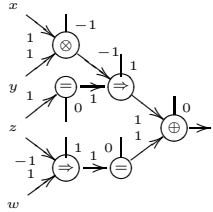
\begin{center}
\tiny
$
\xymatrix @R=6pt @C=8pt { x\ar[dr]_{1} & \ar@{-}[d]^{-1} & \\
                           & *+[o][F-]{\otimes} \ar[rd]^{-1}& \ar@{-}[d]^{1} \\
           y\ar[ur]^{1} & *+[o][F-]{=}\ar[r]_{1}  &  *+[o][F-]{\Rightarrow} \ar[rd]_{1} &\ar@{-}[d]^{0}\\
           z\ar[dr]_{-1} \ar[ur]^{1}& \ar@{-}[d]^{1}\ar@{-}[u]_{0} & \ar@{-}[d]_{0} &*+[o][F-]{\oplus} \ar[r]&\\
                           & *+[o][F-]{\Rightarrow} \ar[r]^{1}&*+[o][F-]{=} \ar[ru]^{1}\\
           w\ar[ur]^{1} &  &  \\
           }
$
\end{center}
\begin{center}
\tiny
\begin{tabular}{lll}
   INTERPRETATION:&\\
   &$(\sum_J\overline{lim}\;D)=\{(x,y,z,w,(x \otimes y)\Rightarrow z)\oplus (z\Rightarrow w)):x,y,z,w\in \Omega\}$\\
\end{tabular}
\end{center}
Note that, we may distinguish, in a neural network,  three types of wires\cite{Michell86}: input wires used to feed the network, hidden wires are those what define links between the first layer of arrows to the least layer of output wires. Let $J\subset I$ be the set of indexes defined by hidden wires, because every multi-morphism defines a map it follows
\[
\sum_J\overline{lim}\;D= \pi_{I\setminus J}\;\overline{lim}\;D,
\]
where $\pi_{I\setminus J}$ is a projection from $\prod_IA_i$ to $\prod_{I\backslash J}A_i$. The map $\sum_J\overline{lim}\;D$ describes the functional dependence between input wires and output wires.

A pattern in a data set can be described by the neural network structure. The network defined by the multi-diagram presented bellow
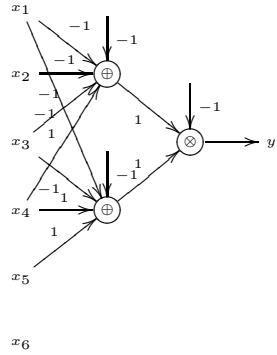
\begin{center}
\tiny
$
\xymatrix @R=15pt @C=20pt {
           x_1\ar[dr]^{-1} \ar[dddr]_{-1}& \ar[d]^{-1} & \\
           x_2\ar[r]^{-1} & *+[o][F-]{\oplus}\ar[rd]_{1}  & \ar[d]^{-1} \\
           x_3\ar[ur]^{-1} \ar[dr]_{-1}&\ar[d]^{-1} & *+[o][F-]{\otimes} \ar[r]& y\\
           x_4 \ar[r]^{1}\ar[ruu]^{1}& *+[o][F-]{\oplus} \ar[ru]^{1}&\\
           x_5\ar[ur]^{1} &  &  \\
           x_6 &  &  \\
           }
$
\end{center}
was generated using the data set of Table \ref{dataset} by the algorithm described in \cite{Leandro09}. For training the such neural networks we changed the Levenderg-Marquardt algorithm \cite{HaganMenhaj99}, restricting the knowledge dissemination in the network structure using soft crystallization \cite{Leandro09}. This procedure reduces neural network plasticity without drastically damaging the learning performance, allowing the emergence of symbolic patterns. This makes the descriptive power of produced neural networks similar to the descriptive power of {\L}ukasiewicz logic propositional language\cite{Castro98}, reducing the information lost on translation between symbolic and connectionist structures. This translation is made using Proposition \ref{conf classification} and can be used as the symbolic description for the network vague limit,
\[
_{\sum_J\overline{lim}\;D=\{(x_1,x_2,x_3,x_4,x_5,x_6,y):\;y=(\neg x_1\oplus \neg x_2 \oplus x_3 \oplus x_4) \otimes (\neg x_1\oplus \neg x_3 \oplus x_4 \oplus x_5)\}}
\]
used as a model for the data \cite{Michell86}. Its quality is evaluated selecting a similarity relation defined in $\Omega^n$. The selection of an adequate similarity is problem dependent. Naturally, prefect descriptions have top degrees of similarity, $\top$, independently of the selected measure.
\end{example}

\begin{table}
\begin{center}
\tiny
\begin{tabular}{|c||c|c|c|c|c|c|c|c|c|c|c|c|c|c|c|c|c|c|c|c|}
  \hline
Att &1&2&3&4&5&6&7&8&9&10&11&12&13&14&15&16&17&18&19&20\\
  \hline
  \hline
$x_1$&0& 0&0.7&0.3&1.0&0.3&0&0.3&0.3&0.7&0.7&0.7&0.7&0.3&0.3&1.0& 0&1.0&0.3&0.3 \\
$x_2$&1.0&0.7& 0&0.7& 0&0.3&.0& 0&1.0& 0&1.0&0.3&0.3&0& 0&0.7& 0&1.0&0.3& 0\\
$x_3$&0.3& 0&0.7&0.7&0.7&1.0&.7&1.0& 0&0.3&1.0&0.7&1.0&0.7&0.7& 0&0.7&0.7& 0& 0\\
$x_4$&1.0&1.0&0.3& 0& 0& 0&0&1.0&1.0&0.3& 0&1.0&0.7&0&0.3&0.3&1.0&0.3&0.7&0.3\\
$x_5$&0.7&0.7& 0&0.7& 0&1.0&.3&0.3& 0&0.7&0.7&1.0&0.7&0& 0&0.3&0.3& 0& 0&1.0\\
$x_6$&0.7&1.0&0.3&0.3&0.7&1.0&.7& 0&0.3&0.3&0.7&1.0& 0&0.3& 0& 0& 0&0.7& 0&1.0\\
\hline
$y$&1.0&1.0&1.0&1.0&0.3&1.0&1.0&1.0&1.0&1.0&1.0&1.0&1.0&1.0&1.0&0.7&1.0&0.7&1.0&1.0\\
  \hline
\end{tabular}
\end{center}
\caption{Data set having 20 cases, with vague propositional variables $x_1,x_2,x_3,x_4,x_5,x_6,y$, in a universe governed by {\L}ukasiewicz logic.}\label{dataset}
\end{table}
\section{Conclusions and future work}
This is a novel approach to vague data modeling, based on extensions of well established  notion used for algetic modeling. Despite of its application to formalize processes of learning using {\L}ukasiewicz neural networks and on vague decision trees, its use for formalizeing generic learning processes seems to be restricted by the nature or our notion of similarity relation. More work must be done in order to use generic kernel function as a mechanism to compare entities.

This working was motivated by the description of a framework to specify vague knowledge bases, having the knowledge described by multi-diagrams. Many real-word application domains are characterized by the presence of both vague and complex relational structure. Research in this fields expanded rapidly in recent years\cite{Domingos06}. There is an increasingly pressing for a unifying framework, a common language for describing and relating the different approaches to statistical relational learning. In this paper we presente our preliminary work on the possibility of applying logic extensions to the established algebraic modeling framework based on sketchs\cite{Ehresm68}. However our approach dificultes the control over the graphic description complexity for concepts. More work must be done for defining  a meta-language to simplify vague description using sketches\cite{Diskin99}.
\bibliographystyle{unsrt}
\bibliography{tesebib2}
\end{document}